\documentclass[11pt]{amsart}
\usepackage[foot]{amsaddr}
\usepackage[margin=1in]{geometry}
\usepackage{bibtex}

\begin{document}

\title{Characterization of permutation gates in the third level of the Clifford hierarchy}

\title[Characterization of permutation gates in $\cC_3$]{Characterization of permutation gates in the third level of the Clifford hierarchy}
\author{Zhiyang He, Luke Robitaille, and Xinyu Tan}
\address{Department of Mathematics, Massachusetts Institute of Technology, Cambridge, MA 02139, USA}
\email{{szhe,lrobitai,norahtan}@mit.edu}

\begin{abstract}
The Clifford hierarchy is a fundamental structure in quantum computation whose mathematical properties are not fully understood.
In this work, we characterize permutation gates---unitaries which permute the $2^n$ basis states---in the third level of the hierarchy. 
We prove that any permutation gate in the third level
must be a product of
Toffoli gates in what we define as \emph{staircase form}, up to left and right multiplications by Clifford permutations. 
We then present necessary and sufficient conditions for a staircase form permutation gate to be in the third level of the Clifford hierarchy.
As a corollary, we construct a family of non-semi-Clifford permutation gates $\{U_k\}_{k\geq 3}$ in staircase form such that each $U_k$ is in the third level but its inverse is \emph{not} in the $k$-th level.
\end{abstract}

\maketitle

\tableofcontents

\section{Introduction}

The Clifford hierarchy is a ubiquitous structure in quantum computation which classifies unitary operations based on their conjugate actions on the Pauli group.
Precisely, the first level $\cC_1$ of the hierarchy $\ch$ is the Pauli group $\cP$, and its subsequent levels are defined recursively: $\cC_k$ is the set of all unitaries $U$ such that $UPU^{\dagger}\in \cC_{k-1}$ for all Pauli operators $P$.
Within $\ch$, gates in the third level are of unique importance to the study of fault-tolerant quantum computation (FTQC)~\cite{shor1995scheme,shor1996fault-tolerant,gottesman1997stabilizer2}, as established by several foundational works. 
The Gottesman-Knill theorem~\cite{gottesman1998knill} states that any circuits made of Pauli and Clifford gates, which are the first two levels of $\ch$, can be efficiently simulated by a classical computer. 
In contrast, adding any non-Clifford gate to the Clifford group forms a universal gate set.
Among the non-Clifford unitaries, gates in $\ch$ can be implemented fault-tolerantly by gate teleportation~\cite{gottesman99teleportation}, where higher-level gates are performed using resource states and lower-level gates.
From this perspective, the non-Clifford gates that are in $\cC_3$ are the ``easiest-to-implement'' non-Clifford gates.
Consequently, $\cC_3$ gates have been at the center of study for FTQC, with decades of extensive research studying their fault-tolerant implementations~\cite{shor1996fault-tolerant,bravyi2005universal2}, synthesis into unitaries~\cite{dawson2005solovay}, algorithmic resource costs~\cite{dalzell2023quantum} and more.

Despite its importance, the mathematical structure of the Clifford hierarchy is not fully understood.
Notably, $\cC_k$ no longer forms a group for any $k\geq 3$, as it is neither closed under multiplication nor closed under inverse.
Ample work has been done to elucidate structures within the hierarchy, which we briefly discuss in \Cref{sec:prior-works}.

In this paper, we study the permutation gates, which are $n$-qubit unitaries that permute  the $2^n$ computational basis states, in $\cC_3$. We denote these gates by $\cC_3^\sym$, and remark that they are both interesting and important for a few reasons. 
First of all, permutation gates correspond to all reversible classical computations on $n$ bits. Gates in $\cC_3$, on the other hand, can be implemented fault-tolerantly using resource states and Clifford gates. 
$\cC_3^\sym$ therefore captures classical gates and computational subroutines which are relatively low-cost to implement for FTQC. 
The canonical example is the Toffoli gate $\tof$, which is classically universal and ubiquitous in quantum circuits.

Unfortunately, products of Toffoli gates (which capture all permutations) are notoriously unruly. For instance, while a single Toffoli gate is in $\cC_3$, a product of as few as two Toffoli gates can leave the hierarchy.\footnote{Specifically, $\tof_{1,2,3}\tof_{1,3,2}$ is not in $\ch$; see equation (E.2) of~\cite{anderson}.} Prior work by Anderson~\cite{anderson} conjectured that $\cC_3^\sym$ are precisely products of pairwise commuting Toffoli gates (up to multiplying on both sides by Clifford permutations), and that $\cC_k^\sym$ is closed under inverse for all $k$. We disprove both of these conjectures in this work.

Furthermore, permutation gates are important components for gates in $\ch$. Beigi and Shor \cite{beigi-shor} proved that all gates in $\cC_3$ are \emph{generalized semi-Clifford}, which means they can be written as $\phi_1\pi d\phi_2$ for Clifford gates $\phi_1,\phi_2$, a diagonal gate $d$, and a permutation gate $\pi$.
The conjecture that all gates in $\ch$ are generalized semi-Clifford remains open~\cite{zeng2008semi-clifford}. 
In~\cite{cui2017diagonal}, Cui, Gottesman, and Krishna fully characterized all the diagonal gates in $\ch$. 
Therefore, characterizing the permutation gates is a crucial step towards characterizing all gates in $\cC_3$ and potentially $\ch$.

In this work, we present a full characterization of $\cC_3^\sym$ (\Cref{rst:staircase-ext}, \Cref{rst:decending-mult}), elucidating an essential substructure of $\cC_3$. 
We then present a family of gates in $\cC_3^\sym$, which disproves Anderson's conjectures and challenges prior understanding of $\cC_3$ (\Cref{rst:Uk-examples}).
Finally, we derive lower bounds on the number of qubits a gate in $\cC_3^\sym$ must be supported on given the ``complexity'' of the function it implements, showing that our construction is optimal (\Cref{rst:deg-k-lowerbound}).

\subsection{Main results and techniques}

To characterize the gates in $\cC_3^\sym$, we study structured products of Toffoli gates.
We define a product of distinct Toffoli gates to be in \emph{staircase form} if each gate $\tof_{i,j,k}$ in the product has $i,j < k$ and the target qubit indices are nondecreasing in the order that the gates are applied. 
For example, the gate in~\Cref{fig:seven_perm_ext} is in staircase form.

\begin{figure}[!ht]
    \centering
    \begin{equation*}
        \begin{quantikz}[slice style=blue] 
        \lstick{$a_1$}&\ctrl{2}&\ctrl{4}&&&&\ctrl{6}&\rstick{$a_1$}\\
        \lstick{$a_2$}&\control{}&&\ctrl{4}&&\ctrl{5}&&\rstick{$a_2$}\\
        \lstick{$a_3$}&\targ{}&&&\ctrl{4}&&&\rstick{$a_3+a_1a_2$}\\
        \lstick{$a_4$}&&\control{}&\control{}&\control{}&&&\rstick{$a_4$}\\
        \lstick{$a_5$}&&\targ{}&&&\control{}&&\rstick{$a_5+a_1a_4$}\\
        \lstick{$a_6$}&&&\targ{}&&&\control{}&\rstick{$a_6+a_2a_4$}\\
        \lstick{$a_7$}&&&&\targ{}&\targ{}&\targ{}&\rstick{$a_7+a_1a_6+a_2a_5+a_3a_4+a_1a_2a_4$}
        \end{quantikz}
    \end{equation*}
    \caption{A permutation gate in staircase form, consisting of $6$ Toffoli gates.}
    \label{fig:seven_perm_ext}
\end{figure}
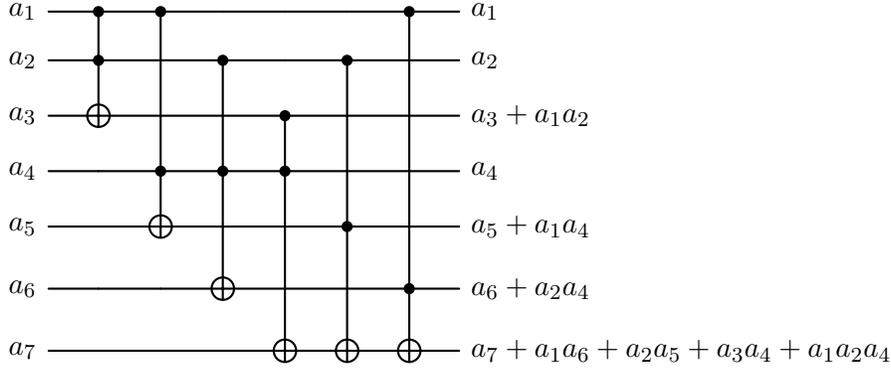

Our first main result, which was presented in an earlier version of this paper~\cite{he2024permutation}, states that every permutation in $\cC_3^\sym$ can be written in staircase form.
\begin{result}[\Cref{thm:c3-staircase}]\label{rst:staircase-ext}
    If $\pi\in \cC_3$ is a permutation gate, then there exist Clifford permutations $\phi_1$, $\phi_2$ and a product $\mu$ of Toffoli gates in \emph{staircase form} such that $\mu\in \cC_3$ and $\pi = \phi_1\mu\phi_2$. 
\end{result}
We remark that there are permutations in staircase form which are not in $\cC_3$. 
To derive a full characterization, we consider bilinear products defined over vectors of $\mathbb{F}_2^n$, which we denote by juxtaposition.
For a commutative and associative bilinear product, we call it a \textit{descending multiplication} if for standard basis vectors $\{e_i\}_{i\in [n]}$, it holds that $e_ie_i = 0$, and $e_ie_j$ (for $i < j$) is in the span of $\{e_k: k> j\}$ (see \Cref{def:descending-mult}).
Our second main result identifies a one-to-one correspondence between gates in $\cC_3^\sym$ and descending multiplications.

\begin{result}[\Cref{thm:perm-is-mult}]\label{rst:decending-mult}
    Every staircase form permutation gate in $\cC_3$ induces a descending multiplication over $\mathbb{F}_2^n$. 
    In correspondence, every descending multiplication over $\mathbb{F}_2^n$ induces a staircase form permutation gate in $\cC_3$.
\end{result}

Using this characterization, we construct a novel family of permutation gates in $\cC_3$.
For each integer $k \geq 3$, we take $n=2^k-1$ and label
each standard basis vector of $\mathbb{F}_2^n$ by $e_S$ for an nonempty subset $S \subseteq [k]$. 
We define a bilinear product operation by setting $e_Se_T = e_{S \cup T}$ if $S \cap T = \varnothing$, and $e_Se_T = 0$ if $S \cap T \neq \varnothing$, and extending linearly.
This gives a descending multiplication, which induces a staircase form $\cC_3$ permutation gate $U_k$.

\begin{result}[\Cref{thm:U_k_non_semi_cliff}]\label{rst:Uk-examples}
    For all $k\geq 3$, we have $U_k\in \cC_3$ but $U_k^{\dagger} \notin \cC_k$. 
\end{result}

We note that the permutation in \Cref{fig:seven_perm_ext} is precisely $U_3$. 
This construction brings new understanding to the study of $\cC_3$ gates in a few ways. First of all, $\{U_k\}_{k\geq 3}$ are the first examples of permutation gates in $\cC_3$ whose inverses leave $\cC_3$ (in fact, leave any fixed level of $\ch$).
As a result, they disprove both conjectures of Anderson~\cite{anderson}. 
Previously, Gottesman and Mochon presented a gate in $\cC_3$ (which is not a permutation; see \Cref{lem:gottesman-mochon} of main text) whose inverse is not in $\cC_3$. 
Our $U_3$ is actually equivalent to their example up to conjugation by a Clifford gate.

The permutation implemented by $U_k$ is also interesting. 
For a permutation $\pi$ on $n$ bits, we represent it by $n$ Boolean polynomials in the $n$ input variables, one polynomial for each output bit.
For example, $\cnot_{1,2}$ maps $(a_1, a_2)\mapsto (a_1, a_2+a_1)$ and $\tof_{1,2,3}$ maps $(a_1,a_2,a_3)\mapsto (a_1,a_2,a_3+a_1a_2)$.
For $U_k$, if we consider qubits $1,2,4,8, \dots, 2^{k-1}$ as ``controls'', and set the other $2^k-1-k$ ``ancilla'' input variables to be zero, then the $2^k-1$ output polynomials representing $U_k$ are precisely all the non-constant monomials of $a_1, a_2, a_4, a_8, \dots, a_{2^{k-1}}$.
In other words, arbitrarily powerful classical computation can be implemented using one $\cC_3$ gate followed by $\cnot$ gates, at the expense of a large number of ancilla qubits.
This extra ancilla cost is, in fact, optimal: using the characterization with descending multiplications, we show the following lower bound.
\begin{result}[\Cref{thm:c3-degree-qubit}]\label{rst:deg-k-lowerbound}
    Any permutation gate in $\cC_3$ with an output bit represented by a polynomial of degree at least $k$ must be supported on at least $2^k-1$ qubits.
\end{result}

\subsection{Future directions}

Combining these four results, our characterization discovers and delineates a novel design space within $\cC_3^\sym$. 
By constructing descending multiplications, which are easier to reason with than Toffoli circuits, we can derive permutation gates in $\cC_3$ which encode different classical computations. 
These computations can then be implemented in FTQC via gate teleportation, where the majority of the cost is offloaded to resource state preparation. 
It would be interesting to explore whether the cost of important FTQC primitives, such as arithmetic, can be reduced with this approach. 

In addition to gate design and teleportation, it would be interesting to explore whether our results can be generalized to higher levels of $\ch$.
Evidently, characterizing more or all of the permutation gates in $\ch$ is consequential to our understanding of $\ch$.
Our results mark a concrete step in this direction.

Another important direction of research is to understand the product of permutation and diagonal gates. 
Significant progress was made in~\cite{anderson}, where Anderson showed several results which, in the case of $\cC_3$, imply that if $\pi d\in\cC_3$, then $\pi\in \cC_3$ and $d\in \ch$. 
Since every gate in $\cC_3$ is generalized semi-Clifford, given our characterization of permutation gates in $\cC_3$ and the characterization of diagonal gates in $\ch$ by~\cite{cui2017diagonal}, can we precisely characterize all of $\cC_3$? 
Such a result would have significant implications for our study of FTQC, given the unique importance of $\cC_3$ gates. 

Finally, our \Cref{rst:deg-k-lowerbound} raises an interesting complexity theory question. Specifically, while $\{U_k\}_{k\geq 3}$ demonstrate that polynomials of arbitrarily high degree can be encoded into a single $\cC_3$ gate, such encodings necessarily incur an ancilla cost exponential in $k$.
In contrast, implementing a degree $k$ monomial in $\cC_{k+1}$ takes only one $k$-controlled NOT gate, supported on $k+1$ qubits. 
How does this tradeoff between polynomial degree and ancilla cost transform through the different levels of $\ch$? 
What implications does this tradeoff have on the fundamental cost of gate teleportation, which is ubiquitous in FTQC?

\subsection{Prior works}\label{sec:prior-works}

Since its introduction by Gottesman and Chuang in 1999 \cite{gottesman99teleportation}, the Clifford hierarchy has been studied by many prior works. 
Early works defined and studied the (generalized) semi-Clifford gates~\cite{zhou2000methodology,dehaene2003clifford,gross2007lu,zeng2008semi-clifford,beigi-shor}.
Subsequent works have elucidated various structures within the qubit Clifford hierarchy~\cite{bengtsson2014order,pllaha2020weyl,rengaswamy2019unifying,anderson,anderson2024controlled,anderson2025affine} and the qudit Clifford hierarchy~\cite{cui2017diagonal,de2021efficient,de2025clifford}.
As the original motivation comes from FTQC, prior works have studied improved gate teleportation protocols for semi-Clifford gates~\cite{zhou2000methodology,zeng2008semi-clifford,de2021efficient}; efficient resource state preparation methods, notably magic state distillation~\cite{bravyi2005universal2,bravyi2012magic-state,campbell2012magic,krishna2019towards,wills2024constant}, for various non-Clifford gates;
and constructions and limitations of quantum codes which admit transversal implementation of gates in different levels of the hierarchy~\cite{bravyi2013classification,anderson2014classification,jochym2018disjointness, hu2021climbingdiagonalcliffordhierarchy,he2025quantum}.

\subsection{Organization}
The content of the paper is divided into sections as follows. 
\Cref{sec-defs} gives background results and lemmas for later use. 
In particular, \Cref{sec-poly} discusses representation of permutation gates as polynomials over $\mathbb{F}_2$. 
In~\Cref{sec-c3-perm}, we present the definition for staircase form, in which all permutations in $\cC_3$ can be written (up to Clifford permutations).
In~\Cref{sec-descending-mult}, we define descending multiplications and prove their one-to-one correspondence with staircase form permutations in \(\cC_3\). 
In~\Cref{sec:family_gates}, we construct our family of non-semi-Clifford gates $\{U_k\}_{k\geq 3}$ where each permutation $U_k$ is in staircase form.
We prove in~\Cref{thm:U_k_non_semi_cliff} that each $U_k\in\cC_3$ but $U_k^\dagger \notin \cC_k$. 
In particular, the smallest example in this family $U_3$ is a $7$-qubit permutation, and it is conjugate to the Gottesman--Mochon example by a Clifford operator. 

In~\Cref{sec-seven-is-best}, we show that $7$ is the smallest number of qubits for which there exists a non-semi-Clifford permutation in $\cC_3$. 
In \Cref{sec-sc-perm}, we classify semi-Clifford permutations and where they appear in the Clifford hierarchy.

\begin{remark}
This work supersedes an earlier version \cite{he2024permutation} by the same authors. 
The prior work established only \Cref{rst:staircase-ext}, which gives a necessary but not sufficient condition for a gate to be in $\cC_3^\sym$. 
By contrast, the present paper provides substantially stronger results.
In particular, we strengthen the condition to be both necessary and sufficient (\Cref{rst:decending-mult}), and as an application, we construct an infinite family of $\cC_3^\sym$ gates whose inverses are not in particular levels of the Clifford hierarchy (\Cref{rst:Uk-examples}).
\end{remark}

\section{Preliminaries}\label{sec-defs}

The single-qubit Pauli gates $I_2$, $X$, $Y$, and $Z$ are given by
\begin{equation*}
    I_2 = \begin{bmatrix} 1&0 \\ 0& 1 \end{bmatrix}, \quad
    X = 
    \begin{bmatrix} 0&1 \\ 1& 0 \end{bmatrix}, \quad 
    Y = \begin{bmatrix} 0& -i \\ i& 0 \end{bmatrix}, \quad 
    Z = \begin{bmatrix} 1&0 \\ 0& -1 \end{bmatrix}. 
\end{equation*}
The \emph{Pauli group} on $n$ qubits, 
denoted as $\cP_n$, 
is the collection of all gates of the form $cP_1 \otimes P_2 \otimes \dots \otimes P_n$ for $c \in \{\pm 1, \pm i\}$ and single-qubit gates $P_1, \ldots, P_n \in \{I_2, X, Y, Z\}$. 
In particular, we denote the set of all the $n$-qubit Pauli $X$ operators as $\cX = \{I_2, X\}^{\otimes n}$.

We frequently use the \emph{Hadamard}, \emph{controlled} NOT, and \emph{Toffoli} gates throughout the paper: 
\begin{equation*}
    H = \frac{1}{\sqrt{2}} \begin{bmatrix} 1&1 \\ 1& -1 \end{bmatrix}, \  
    \cnot_{1,2} = \begin{bmatrix}
        I_2 & \\
        & X
    \end{bmatrix}, \
    \tof_{1,2,3} = \begin{bmatrix}
        I_6 & \\
        & X
    \end{bmatrix}. 
\end{equation*}

We can view the action of CNOT as $|a_1 \ra \otimes |a_2 \ra \mapsto |a_1 \ra \otimes |a_1+a_2 \ra$ where the first qubit is the \emph{control} and the second qubits is the \emph{target}.  
Similarly, we can view the Toffoli gate as $|a_1 \ra \otimes |a_2 \ra \otimes |a_3 \ra \mapsto |a_1 \ra \otimes |a_2 \ra \otimes |a_3+a_1a_2 \ra$ where the first two qubits are controls and the third is the target.
We use subscripts to denote the qubits that a gate acts upon. 
For example, $Y_4$ is a Pauli $Y$ gate acting on the fourth qubit, and $\cnot_{3,1}$ is a CNOT gate with the third qubit as control and the first qubit as target.

The \emph{Clifford group} on $n$ qubits is the normalizer of the Pauli group in the unitary group. 
It can be generated by the Pauli group, the Hadamard and phase gate on each qubit, the CNOT gate on each ordered pair of distinct qubits, and $\{cI : |c|=1 \}$. 
Henceforth we refer to elements of $\{cI: |c|=1\}$ as \emph{phases} (not to be confused with the phase gate).

\subsection{Gates in the Clifford hierarchy}\label{sec:prelim-CH}
The Clifford hierarchy is defined recursively, with the first level being the Pauli gates.
\begin{definition}[The Clifford hierarchy]
Let $n$ be the number of qubits. 
Let $\cC_1 = \cP_n$. 
For $k \geq 2$, inductively define $\cC_k$ to be the set of all unitaries $U$ such that $UPU^{\dagger} \in \cC_{k-1}$ for all $P \in \cP_n$. 
Note that $\cC_2$ is the Clifford group. 
The set $\mathcal{CH} :=\cC_1 \cup \cC_2 \cup \cC_3 \cup \dots$ is called the \emph{Clifford hierarchy}; we refer to $\cC_k$ as the $k$-th layer of $\mathcal{CH}$.
\end{definition}

We list a few standard facts of the Clifford hierarchy.

\begin{fact}\label{prop:cliff_hier}
\phantom{...}
    \begin{enumerate}
        \item For any $k$, $\cC_k$ is finite up to phase and $\cC_k\subseteq \cC_{k+1}$. 
        \item For $k\geq 2$, $\cC_k$ is closed under left and right multiplication of Clifford gates.
        \item For $k \geq 3$, $\cC_k$ is not a group.
        \item For any $k$, $\cC_k$ is closed under complex conjugation.\label{prop:closed_conj}
    \end{enumerate}
\end{fact}

We say that a gate is a \emph{permutation gate} if it corresponds to a $2^n \times 2^n$ permutation matrix. 
Note that this is different from only permuting the qubits. 
Note that $\cX$ is exactly the set of all permutation gates in $\cP_n$. 
A gate is called \emph{diagonal} if its associated matrix is diagonal.

\begin{definition}[Semi-Clifford and generalized semi-Clifford gates]\label{def:sc-gsc}
A gate is \emph{semi-Clifford} if it can be written as $\phi_1d\phi_2$ for some Clifford gates $\phi_1, \phi_2$ and a diagonal gate $d$.
A gate is \emph{generalized semi-Clifford} if it can be written as $\phi_1\pi d\phi_2$ for some Clifford gates $\phi_1, \phi_2$, a permutation gate $\pi$, and a diagonal gate $d$.
\end{definition}

Observe that the inverse of a semi-Clifford gate is semi-Clifford. The inverse of a generalized semi-Clifford gate is generalized semi-Clifford, as we can write $(\phi_1 \pi d\phi_2)^{-1} = \phi_2^{-1} \pi^{-1} (\pi d^{-1} \pi^{-1}) \phi_1^{-1}$, and $\pi d^{-1} \pi^{-1}$ is diagonal. If we multiply a semi-Clifford (resp. generalized semi-Clifford) element on the left or right by a Clifford gate, the resulting operator is still semi-Clifford (resp. generalized semi-Clifford).

For a maximal abelian subgroup $A$ of $\cP_n$, let $\SPAN(A)$ denote its linear span with complex coefficients. 
The following lemma shows equivalent definitions of (generalized) semi-Clifford gates.
\begin{lemma}\label{lem:sc-with-subgps}
An operator $U$ is semi-Clifford if and only if there exist maximal abelian subgroups $A_1$ and $A_2$ of $\cP_n$ such that $UA_1U^{\dagger} = A_2$. An operator $U$ is generalized semi-Clifford if and only if there exist maximal abelian subgroups $A_1$ and $A_2$ of $\cP_n$ such that $U\SPAN(A_1)U^{\dagger} = \SPAN(A_2)$.
\end{lemma}
We note that in literature, semi-Clifford and generalized semi-Clifford are usually defined as in~\Cref{lem:sc-with-subgps} whereas~\Cref{def:sc-gsc} is proved as a proposition~\cite{dehaene2003clifford,zeng2008semi-clifford,beigi-shor,anderson}.
For proofs of this equivalence, we refer readers to Appendix~A of~\cite{anderson}.

\begin{proposition}[Semi-Clifford gates are closed under taking inverses]\label{lem:respect-sc}
For any $k$, the inverse of any semi-Clifford element of $\cC_k$ is in $\cC_k$.
\end{proposition}

\begin{proof}
For any $U \in \cC_k$ that is semi-Clifford, 
by~\Cref{def:sc-gsc}, we can write $U = \phi_1 d \phi_2$ for some Clifford gates $\phi_1, \phi_2$ and a diagonal gate $d$. 
Using~\Cref{prop:cliff_hier} repeatedly,  
we know that $d = \phi_1^{-1} U \phi_2^{-1} \in \cC_k$. Hence, $d^{-1} = d^{\dagger} = \overline{d} \in \cC_k$ and thus $U^{-1} = \phi_2^{-1} d^{-1} \phi_1^{-1} \in \cC_k$.
\end{proof}

It was conjectured in~\cite{zeng2008semi-clifford} that all gates in $\cC_3$ are semi-Clifford.
This was disproved by Gottesman and Mochon~\cite{beigi-shor} with the following $7$-qubit unitary. 

\begin{lemma} \label{lem:gottesman-mochon}
For $n=7$, $\cC_3$ contains a non-semi-Clifford element.
\end{lemma}

\begin{proof}
Let $G$ be the $7$-qubit gate given by
\begin{equation*}
    G = \cswap_{7,1,6}\cswap_{7,2,5}\cswap_{7,4,3} \cdot \ccz_{1,2,3}\ccz_{1,4,5}\ccz_{2,4,6}\ccz_{3,5,6}, 
\end{equation*}
where CSWAP denotes the controlled SWAP gate and CCZ denotes the controlled controlled $Z$ gate. 
See~\Cref{fig:circ_gottesman_Mochon} for a circuit diagram. 

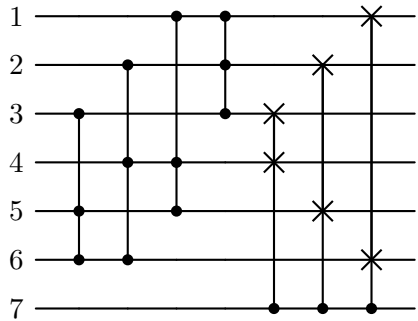
\begin{figure}[!ht]
    \centering
    \begin{equation*}
        \begin{quantikz}[slice style=blue] 
        \lstick{$1$}&&&\ctrl{4}&\ctrl{2}&&&\swap{5}&\\
        \lstick{$2$}&&\ctrl{4}&&\control{}&&\swap{3}&&\\
        \lstick{$3$}&\ctrl{3}&&&\control{}&\swap{1}&&&\\
        \lstick{$4$}&&\control{}&\control{}&&\targX{}&&&\\
        \lstick{$5$}&\control{}&&\control{}&&&\targX{}&&\\
        \lstick{$6$}&\control{}&\control{}&&&&&\targX{}&\\
        \lstick{$7$}&&&&&\ctrl{-4}&\ctrl{-5}&\ctrl{-6}&
        \end{quantikz}
    \end{equation*}
    \caption{Circuit diagram for the Gottesman--Mochon seven-qubit gate $G$ (with time flowing from left to right).}
    \label{fig:circ_gottesman_Mochon}
\end{figure}

It can be verified with a computer program that $G \in \cC_3$. If $G$ were semi-Clifford, then we would have $G^{-1} \in \cC_3$ by~\Cref{lem:respect-sc}. However, a computer calculation shows that $G^{-1} \notin \cC_3$ (in particular, $G^{-1}X_7G \notin \cC_2$). Thus, $G$ is not semi-Clifford.
\end{proof}

This $7$-qubit operator is the smallest known example of a non-semi-Clifford operator in $\cC_3$. 
\cite{zeng2008semi-clifford} showed that for $n \leq 3$, all elements of $\cC_3$ are semi-Clifford. 
Recently, \cite{anderson2025affine} showed that for $n=4$, all elements of $\cC_3$ are semi-Clifford. 
For $n = 5$ or $6$, it is an open problem whether there is a $\cC_3$ operator that is non-semi-Clifford. We make partial progress on this problem in \Cref{sec-seven-is-best} by showing that all permutation gates in $\cC_3$ on at most six qubits are semi-Clifford.

While the semi-Clifford characterization of $\cC_3$ was disproved, Beigi and Shor~\cite{beigi-shor} proved that every gate in $\cC_3$ is generalized semi-Clifford.

\begin{theorem} \label{c3-is-gsc}
Every element of $\cC_3$ is generalized semi-Clifford.
\end{theorem}

The conjecture that every gate in $\ch$ is generalized semi-Clifford~\cite{zeng2008semi-clifford} remains open. Partial progress was made in~\cite{anderson}. 

\begin{conjecture} \label{hierarchy-is-gsc}
Every element of the Clifford hierarchy is generalized semi-Clifford.
\end{conjecture}

\begin{remark}
    A generalized semi-Clifford gate takes the form $\phi_1 \pi d\phi_2$.
    A similar form of $\pi d \phi$ is considered in the context of approximate unitary designs or pseudorandom unitaries in~\cite{metger2024simple,chen2024incompressibility}, where $\phi$ and $\pi$ are sampled uniformly at random from their respective groups, and $d$ is a diagonal gate with random $\pm 1$ entries. 
\end{remark}

\subsection{Polynomial representations of permutations} \label{sec-poly}
To study the permutation gates in the Clifford hierarchy, we represent them as collections of Boolean polynomials.

\begin{fact}\label{poly-rep}
Any function from $\mathbb{F}_2^n$ to $\mathbb{F}_2$ can be uniquely written as an $n$-variable polynomial that has degree at most $1$ in each variable.
\end{fact}

\begin{lemma}\label{diagonal-in-ch}
For any function $f: \mathbb{F}_2^n \rightarrow \mathbb{F}_2$, the diagonal gate $\sum_{a\in\bF_2^n} (-1)^{f(a)} \ketbra{a}$ is in $\cC_k$ if and only if $f$, considered as a polynomial, has degree at most $k$.
\end{lemma}

\begin{proof}
This is a special case of the main theorem in~\cite[See Eq. (1)]{cui2017diagonal}. 
\end{proof}

\begin{definition}[Polynomial representation]~\label{def:polynomial_rep}
    Given a permutation gate $\pi: \mathbb{F}_2^n\rightarrow \mathbb{F}_2^n$, let $\pi_i: \mathbb{F}_2^n \rightarrow \mathbb{F}_2$ denote the function $\pi$ restricted to the $i$-th output bit, i.e.\
    \begin{equation*}
        \pi = \sum_{a\in\bF_2^n} \ketbratwo{\pi_1(a),\ldots, \pi_n(a)}{a}.
    \end{equation*}
    From~\Cref{poly-rep} we know that each $\pi_i$ can be written as a polynomial in the input bits. 
    We refer to $(\pi_1, \ldots, \pi_n)$ as the \emph{polynomial representation of $\pi$} and $\pi_i$ as the \emph{$i$-th coordinate of $\pi$}. 
\end{definition}
As an example, $\tof_{1,2,3}$ can be represented as $(a_1,a_2,a_3) \mapsto (a_1,a_2,a_3+a_1a_2)$.
We prove a few useful lemmas regarding the polynomial representations of permutation gates.

\begin{lemma} \label{ck-perm-poly}
For any integer $k\geq 1$ and permutation gate $\pi \in \cC_{k+1}$, each coordinate of $\pi^{-1}$ has degree at most $k$.
\end{lemma}
\begin{proof}
For each $i\in [n]$, we have 
\begin{equation}\label{eq:pi_conjugate_on_Z}
    \cC_k \ni \pi Z_i \pi^{-1} = \sum_{a\in\bF_2^n} (-1)^{a_i} \ketbra{\pi(a)} = \sum_{a\in\bF_2^n} (-1)^{(\pi^{-1}(a))_i} \ketbra{a} =  \sum_{a\in\bF_2^n} (-1)^{\pi_i^{-1}(a)} \ketbra{a}.
\end{equation}
It follows from~\Cref{diagonal-in-ch} that $\pi_i^{-1}$, the $i$-th coordinate of $\pi^{-1}$, must have degree at most $k$. 
\end{proof}

\begin{remark} \label{non-quad-perm}
For $\pi \in \cC_3$, \Cref{ck-perm-poly} tells us that every coordinate of $\pi^{-1}$ has degree at most $2$; however, as we will see in~\Cref{sec:family_gates}, the coordinates of $\pi$ themselves do not necessarily have degree at most $2$.
\end{remark}

We denote by $e_1, \ldots, e_n$ the standard basis of $\mathbb{F}_2^n$.
\begin{proposition}[Clifford permutations] \label{clifford-perm}
For any $n \times n$ invertible matrix $M$ over $\mathbb{F}_2$ and any vector $w$, the permutation gate sending $|v \ra \mapsto |Mv + w\ra$ is Clifford. Conversely, any Clifford permutation is of this form for some $M$ and $w$.
\end{proposition}
\begin{proof} 
For the first claim, it is clear that $|v \ra \mapsto |Mv+w \ra$ is a permutation, which we denote $\pi$.
To show that $\pi$ is Clifford, it suffices to show that $\pi X_i \pi^{-1}, \pi Z_i \pi^{-1} \in \cP_n$. 
For $\pi X_i \pi^{-1}$, it sends 
\[
|v \ra \mapsto |M^{-1}(v-w) \ra \mapsto |M^{-1}(v-w)+e_i \ra \mapsto |M(M^{-1}(v-w)+e_i) +w \ra = |v+Me_i \ra.
\] 
Therefore $\pi X_i \pi^{-1}$ is equivalent to a product of $X$ gates. 
For $\pi Z_i \pi^{-1}$, it sends 
\[
|v \ra \mapsto |M^{-1}(v-w) \ra \mapsto (-1)^{e_i^\top M^{-1}(v-w)} |M^{-1}(v-w) \ra \mapsto (-1)^{e_i^\top M^{-1}(v-w)} |v \ra. 
\] 
We can rewrite this as $|v \ra \mapsto (-1)^{-e_i^\top M^{-1}w}(-1)^{((M^{-1})^\top e_i)^\top v} |v \ra$, so this is a product of $Z$ operators up to a phase of $\pm 1$. 
Hence we have $\pi \in \cC_2$.

For the converse claim, we have $\pi^{-1}$ is a permutation in $\cC_2$ (as $\cC_2$ is a group). 
Using~\Cref{ck-perm-poly} with $k=1$, every coordinate of $(\pi^{-1})^{-1} = \pi$ has degree at most $1$.
This directly yields 
a matrix $M$ and vector $w$ so that $\pi$ can be written as $|v \ra \mapsto |Mv+w \ra$. 
Since $\pi$ is a permutation, $M$ must be invertible. 
\end{proof}

Recall that $\cX$ is the set of all the $n$-qubit Pauli $X$ operators.
\begin{proposition}[Pauli permutations] \label{conj-cliff-perm}
Suppose $X'_1, X'_2, \ldots, X'_m\in \cX$ are independent (that is, no nontrivial product of them yields the identity). 
Then there exists some Clifford permutation $\nu$ such that $\nu |0^n \ra = |0^n \ra$ and $\nu X_i \nu^{-1}=X'_i$ for all $i\in [m]$.
\end{proposition}
\begin{proof}
Note that we can view each $X'_i$ as a map $|v \ra \mapsto |v + v_i \ra$. 
The independence property gives that $v_1, \ldots, v_m$ are linearly independent, i.e.\ there exists an invertible matrix $M$ such that $Me_i=v_i$ for all $i\in [m]$. 
Let $\nu$ be $|v \ra \mapsto |Mv \ra$ which is a Clifford permutation by~\Cref{clifford-perm}, and has $\nu |0^n \ra = |0^n \ra$. 
Then $\nu X_i \nu^{-1}$ sends \[ |v \ra \mapsto |M^{-1}v \ra \mapsto |M^{-1}v+e_i \ra \mapsto |M(M^{-1}v+e_i) \ra = |v + Me_i \ra = |v+v_i \ra, \]
which means that $\nu X_i \nu^{-1} = X'_i$, as desired.
\end{proof}

\subsection{Anderson's conjectures}

Besides polynomials, we consider the more operational representation of permutation gates as products of multi-controlled NOT gates, which we denote by $C^kX$ for $k\geq 0$ (note that Toffoli is $C^2X$). 
In~\cite{anderson}, Anderson considered \emph{mismatch-free} circuits, which are products of pairwise commuting multi-controlled NOT gates.
\begin{theorem}[Theorem~D.4 of~\cite{anderson}]\label{thm:anderson-mismatch-free}
    A mismatch-free permutation circuit is in $\ch$ at the level of the highest-level gate in the circuit.
\end{theorem}
\noindent Following this classification, Anderson gave two conjectures on permutation gates in $\ch$.
\begin{conjecture}\label{conj:c3-tof}
    A permutation is in $\cC_3$ if and only if it can be written as a circuit of commuting Toffoli gates, up to left and right multiplication by Clifford gates.
\end{conjecture}
\begin{conjecture}\label{conj:perm-inverse}
    For a permutation $\pi\in \cC_k$, we have $\pi^\dagger\in \cC_k$. 
\end{conjecture}
In \Cref{sec-sc-perm}, we show that mismatch-free permutation circuits are precisely the semi-Clifford permutation gates (\Cref{thm:semi-clifford-mismatch-free}). 
However, as we will see in the next few sections, the permutation gates in $\cC_3$ form a much richer space, which we characterize in this work. In particular, we disprove both \Cref{conj:c3-tof} and \Cref{conj:perm-inverse}.

\section{Staircase form representations of \texorpdfstring{$\cC_3$}{C₃} permutations} 
\label{sec-c3-perm}

We begin by presenting the most important definition of our work. 

\begin{definition}[Staircase form Toffoli circuits]\label{def:staircase}
    A product of pairwise distinct Toffoli gates is said to be in \emph{staircase form} if each gate $\tof_{i,j,k}$ in the product has $i<j < k$ 
    and the target qubits are in nondecreasing order in the order that the gates are applied.
    See~\Cref{fig:eg_staircase} for an example.
\end{definition}
\begin{figure}[!ht]
    \centering
    \begin{equation*}
        \begin{quantikz}[slice style=blue] 
        \lstick{$a_1$}&\ctrl{2}&\ctrl{3}&\ctrl{3}&\rstick{$a_1$}\\
        \lstick{$a_2$}&\control{}&&\control{}&\rstick{$a_2$}\\
        \lstick{$a_3$}&\targ{}&\control{}&&\rstick{$a_3+a_1a_2$}\\
        \lstick{$a_4$}&&\targ{}&\targ{}&\rstick{$a_4+a_1a_3$}
        \end{quantikz}
    \end{equation*}
    \caption{This circuit for  $\tof_{1,2,4}\tof_{1,3,4}\tof_{1,2,3}$ is in staircase form but not mismatch-free, as qubit $3$ is used as a control for $\tof_{1,3,4}$ and a target for $\tof_{1,2,3}$.}
    \label{fig:eg_staircase}
\end{figure}
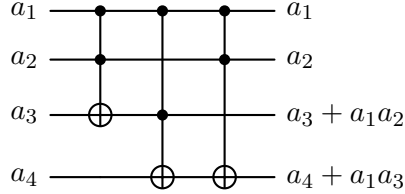

Note that a staircase form Toffoli circuit is unique up to ordering of gates with the same target.
We say that a permutation is in staircase form if it can be written as a Toffoli circuit in staircase form.

Our main result of this section is the following.

\begin{restatable}{theorem}{mainthmstaircase}
\label{thm:c3-staircase}
Suppose $\pi \in \cC_3$ is a permutation gate. Then there exist Clifford permutations $\phi_1, \phi_2$ and a staircase form permutation $\mu \in \cC_3$ such that $\pi = \phi_1\mu\phi_2$.
\end{restatable} 

Since $\pi\in\cC_3$ is a permutation, each $\pi X_j\pi^{-1}$ is a Clifford permutation. 
Then by \Cref{clifford-perm}, we can find a binary matrix $A_j$ and a vector $b_j$ such that $\pi X_j\pi^{-1}$ implements the permutation $|v \ra \mapsto |v + A_jv+b_j \ra$. 
To prove \Cref{thm:c3-staircase}, we construct a sequence of Clifford operators to reduce $A_j, b_j$ to a specific form, which will help us build a staircase form representation of $\pi$.

We caution that \Cref{thm:c3-staircase} is not if-and-only-if, because there exist permutations in staircase form that are not in $\cC_3$.
For example, $\pi'=\tof_{3,4,5}\tof_{1,2,3}$ is in staircase form, but $\pi' X_1\pi'^{-1} = X_1 \cnot_{2,3}\tof_{2,4,5}\notin \cC_2$.

\subsection{Reducing \texorpdfstring{$\cC_3$}{C₃} permutations}

The goal of this subsection is to reduce a $\cC_3$ permutation to a much more friendly form, which serves as the first step in proving \Cref{thm:c3-staircase}: 
\begin{restatable}[Reducing $\cC_3$ permutations]{lemma}{lemcleanpi}\label{lem:clean_pi}
    For any permutation gate $\pi\in\cC_3$, there exist Clifford permutations $\phi_1,\phi_2$ such that the following is true. Let $\tau = \phi_1 \pi \phi_2$. Then 
    $\tau\ket{0^n} = \ket{0^n}$ and 
    for each $i$, we have
    $\tau \ket{e_i} = \ket{e_i}$ and
    $\tau X_i \tau^{-1}$ is given by $|v \ra \mapsto |v + A_i v+e_i \ra$ for some strictly lower triangular matrix $A_i$ over $\mathbb{F}_2$. 
\end{restatable}

To prove \Cref{lem:clean_pi}, we will use a series of helper lemmas which can be proved via standard linear algebra arguments. 

The first helper lemma stated below is essentially the same as standard results on simultaneous triangularization of commuting nilpotent matrices; see, for example, \cite{radjavi2000simultaneous}. We include a proof for completeness. 
\begin{lemma} \label{simul-lower-tri}
Suppose $A_1, \ldots, A_k$ are linear transformations of an $n$-dimensional vector space $V$ over a field $F$ such that $A_j^2=0$ and $A_iA_j=A_jA_i$. Then there exists a basis of $F^n$ in which all the $A_i$ are strictly lower triangular. Recall that a matrix is \emph{strictly lower triangular} if it is lower triangular, and all diagonal elements are $0$.
\end{lemma}

\begin{proof} 
First we show that the intersections of kernels of $A_i$, namely $\cap_{i}\ker(A_i)$, is non-empty. 
Assume for the sake of contradiction it is empty, and let $v$ be a non-zero vector which maximizes the number of indices $i$ for which $A_iv=0$. 
Take $j$ with $A_jv \neq 0$, we see that $A_i(A_jv) = A_jA_iv = 0$ for any $i$ with $A_iv = 0$, and $A_j(A_jv) = A_j^2v = 0$. 
Therefore, $A_jv$ is in more kernels $\ker(A_i)$ than $v$ is, contradicting our assumption on $v$. 
Hence, there must exist non-zero $v$ such that for all $i$, $A_iv = 0$.

We now induct on $n$. 
Consider the $(n-1)$-dimensional vector space $V/\{v\}$. 
Since $v\in \cap_{i}\ker(A_i)$, all linear transformations $A_i$ are well-defined on $V/ \{v\}$ and satisfy $A_iA_j = A_jA_i$ and $A_i^2=0$. 
So there exists a basis $v_1+\{v\}, \ldots, v_{n-1}+\{v\}$ of $V/\{v\}$ in which all the $A_i$ are strictly lower triangular. 
Now take the basis $v_1, v_2, \ldots, v_{n-1}, v$ (in that order) on $V$, one can check that all $A_i$ are strictly lower triangular, as desired.
\end{proof}

For any nonzero column vector $v$ over $\mathbb{F}_2$, let $\alpha(v)$ denote the index of its first nonzero component. Set $\alpha(0) = \infty$ as convention.

\begin{lemma} \label{lower-tri-mult}
Suppose $A$ is an $n \times n$ strictly lower triangular matrix over $\mathbb{F}_2$, and $b$ is a nonzero column vector in $\mathbb{F}_2^n$. Then $\alpha(Ab)>\alpha(b)$. 
\end{lemma}

\begin{proof}
This follows directly from the definition of strictly lower triangular.
\end{proof}

\Cref{lower-tri-mult} will be used tacitly throughout what follows.

\begin{proposition} \label{twisted-gauss-elim}
Suppose we have a list of tuples $(A_1, b_1), \ldots, (A_n, b_n)$, where each $A_i$ is an $n \times n$ strictly lower triangular matrix over $\mathbb{F}_2$ and each $b_i$ is a column vector in $\mathbb{F}_2^n$. Suppose we can perform the following operations:

\begin{enumerate}

\item ``Swap'': swap the indices of two pairs $(A_i, b_i)$ and $(A_j, b_j)$, or

\item ``Compose'': choose two distinct indices $i$ and $j$, and update $A_i$ to be $A_i+A_j+A_iA_j$ and update $b_i$ to be $b_i+b_j+A_ib_j$.
\end{enumerate}
Then it is possible to perform operations either to reach a state where $b_i=e_i$ for all $i$, or to reach a state where some $b_i$ is $0$. 
\end{proposition}

\begin{proof}
First note that the new matrix given by ``compose'' is always strictly lower triangular.
Let us assume without loss of generality that we cannot reach $b_i=0$ for any $i$. 
We describe a two-phase procedure which will reach the state $b_i=e_i$ for all $i$.

For the first phase of the process, we will reach a state with $\alpha(b_i)=i$ for all $i$, as follows. There are finitely many reachable states, so we can reach a state maximizing the value of $\sum_{i=1}^n \alpha(b_i)$ over all reachable states. In this state, the values of $\alpha(b_i)$ must be pairwise distinct. To see this, suppose $\alpha(b_i) = \alpha(b_j) = k$ for some $i \neq j$. Then note that $\alpha(b_i+b_j) > k$ and $\alpha(A_ib_j) > k$, so $\alpha(b_i+b_j+A_ib_j) > k = \alpha(b_i)$. 
This means if we compose $(A_i, b_i)$ with $(A_j, b_j)$ to obtain $(A_i+A_j+A_iA_j, b_i+b_j+A_ib_j)$, we will increase the value of $\sum_{i=1}^n \alpha(b_i)$, which is a contradiction. 
Therefore $\alpha(b_1), \ldots, \alpha(b_n)$ are pairwise distinct, so they must equal $1,2, \ldots, n$ in some order. Perform swaps so that $\alpha(b_i)=i$ for all $i$, this completes the first phase.

The second phase of our procedure is simply row reduction. Suppose there exists $b_i\ne e_i$ and let $\alpha(b_i + e_i) = k > i$. Then we can compose $(A_i, b_i)$ with $(A_k, b_k)$ to get the new vector $b_i + b_k + A_ib_k$. Observe that 
\[
\alpha(b_i + e_i + b_k) > k, \alpha(A_ib_k) > k \Rightarrow \alpha(b_i + b_k + A_ib_k + e_i) > k. 
\]
Therefore we can repeat this procedure until $\alpha(b_i + e_i) > n$, which means $b_i = e_i$. Repeating this for all $i$ leads to our desired state.
\end{proof}

Using \Cref{clifford-perm,simul-lower-tri,twisted-gauss-elim,conj-cliff-perm}, we are ready to prove \Cref{lem:clean_pi}, restated below for convenience. 

\lemcleanpi*
\begin{proof}
By multiplying $\pi$ by suitable $X$'s on the left, we assume without loss of generality that $\pi|0^n \ra = |0^n \ra$.

Since $\pi\in\cC_3$ is a permutation, each $\pi X_j \pi^{-1}$ is a Clifford permutation. Then by \Cref{clifford-perm} we can  
write $\pi X_j \pi^{-1}$ as $|v \ra \mapsto |v + A_jv+b_j \ra$ for some matrix $A_j$ and vector $b_j$ over $\mathbb{F}_2$. 
Since $X_j^2 = I$ and $X_iX_j = X_jX_i$, we have $A_j^2 = 0$ and $A_iA_j = A_jA_i$. 
By~\Cref{simul-lower-tri}, these conditions imply that there is some basis in which the $A_j$ are simultaneously strictly lower triangular, 
so we can take some matrix $M$ such that, for all $i$, $MA_iM^{-1}$ is strictly lower triangular. 
Let $\psi$ be the permutation gate $|v \ra \mapsto |Mv \ra$, which is Clifford by~\Cref{clifford-perm}. 
Now $\psi \pi |0^n \ra = |0^n \ra$, and the map $(\psi \pi) X_j (\psi \pi)^{-1}$ sends 
\[
|v \ra \mapsto |M^{-1}v \ra \mapsto |M^{-1}v + A_jM^{-1}v + b_j \ra \mapsto |v + MA_jM^{-1}v + Mb_j \ra. 
\] 
Therefore, by replacing $\pi$ with $\psi \pi$, we can assume without loss of generality that all matrices $A_j$ are strictly lower triangular, and preserve the property that $\pi |0^n \ra = |0^n \ra$.

We now apply~\Cref{twisted-gauss-elim} to reduce $b_i$ to $e_i$. Note that the map 
\begin{equation*}
    \pi X_iX_j\pi^{-1} = (\pi X_i\pi^{-1})(\pi X_j\pi^{-1})
\end{equation*}
sends
\begin{align*}
|v \ra \mapsto |v+A_jv+b_j \ra \mapsto& |(v+A_jv+b_j) + A_i(v+A_jv+b_j)+b_i \ra \\ =& |v+(A_i+A_j+A_iA_j)v+b_i+b_j+A_ib_j \ra,
\end{align*}
which corresponds to the compose operation. 
Therefore, by~\Cref{twisted-gauss-elim}, there exists a sequence of swaps and multiplications which transform the generators $ X_1, \ldots, X_n $ to $ X_1', \ldots, X_n'$, where each $X_i'$ is a product of $X$ gates, 
such that either $\pi X_i'\pi^{-1}$ sends $|v \ra \mapsto |v+A_i'v+e_i \ra$ for all $i$, or there exists $i$ such that $\pi X_i'\pi^{-1}$ sends $|v \ra \mapsto |v+A_i'v\ra$. 
However, the latter case cannot happen, as otherwise $\pi X_i'\pi^{-1}$ sends $\ket{0^n}\mapsto \ket{0^n + A_i'0^n} = \ket{0^n}$, which contradicts $\pi\ket{0^n} = \ket{0^n}$.

Since $X'_1, \ldots, X'_n$ form a basis for $\cX$, by~\Cref{conj-cliff-perm}, there exists a Clifford permutation $\nu$ such that $\nu |0^n \ra = |0^n \ra$ and $\nu X_i \nu^{-1} = X'_i$ for all $i$. Note that $\pi \nu |0^n \ra = |0^n \ra$. Therefore, if we replace $\pi$ with $\pi \nu$, we get $(\pi \nu) X_i (\pi \nu)^{-1} = \pi X'_i \pi^{-1}$, and we preserve $\pi |0^n \ra = |0^n \ra$, which means we can assume without loss of generality that $b_i=e_i$ for all $i$. 
In particular, we have $\pi|e_i \ra = (\pi X_i \pi^{-1}) |0^n \ra = |e_i \ra$. 
\end{proof}

\subsection{Proof of \texorpdfstring{\Cref{thm:c3-staircase}}{Theorem 3.2}}

Now we only need one more ingredient to prove \Cref{thm:c3-staircase}: the polynomial representations of staircase form permutations, whose proof follows intuitively from the circuit.

\begin{lemma}[Characterization of staircase form permutation via the polynomial representation of its inverse]\label{staircase-poly}
A permutation gate $\pi$ is staircase form if and only if, in the polynomial representation of $\pi^{-1}$, for all $k$, the $k$-th coordinate is $a_k$ plus a (possibly empty) sum of terms of the form $a_ia_j$ with $i<j<k$. 

Furthermore, given a permutation $\pi$ written as a staircase form Toffoli circuit, for any $i<j<k$, we have that $\tof_{i,j,k}$ appears in the product if and only if the $k$-th coordinate in the polynomial form of $\pi^{-1}$ contains an $a_ia_j$ term.
\end{lemma}

\begin{proof}
If $\pi$ is a product of Toffoli gates in staircase form, then $\pi^{-1}$ is the product of those same Toffoli gates in reverse order.
In that product, whenever a gate is applied, its controls have never been targeted so far, and thus are unchanged from the input. This means that, when performing the gates of $\pi^{-1}$ in order, a gate $\tof_{i,j,k}$ adds a term of $a_ia_j$ to the $k$-th coordinate of the output. All parts of the desired result now can be easily shown.
\end{proof}

We are now ready to prove \Cref{thm:c3-staircase}, restated below for convenience. 

\mainthmstaircase*
\begin{proof}
By \Cref{lem:clean_pi}, we can assume without loss of generality that $\pi \ket{0^n} = \ket{0^n}$ and for each $i\in [n]$, 
$\pi \ket{e_i} = \ket{e_i}$ and $\pi X_i \pi^{-1}$ is given by $|v \ra \mapsto |v + A_i v+e_i \ra$ for some strictly lower triangular matrix $A_i$ over $\mathbb{F}_2$.

We now show that for any $v$, if $\pi |v \ra = |w \ra$, then $\alpha(v)=\alpha(w)$. Suppose for the sake of contradiction that this is false. We know it is true for $v=0^n$ or $e_n$, so $\alpha(v) < n$ in any counterexample. Take the largest $k$ for which there exists a counterexample with $\alpha(v)=k$. We know $\pi|e_k \ra= |e_k \ra$, so $v \neq e_k$. Let $u\ne 0^n$ be such that $\pi |v+e_k \ra = |u \ra$. Then $\alpha(v+e_k) > k$, so $\alpha(u) = \alpha(v+e_k)= m$ for some $m > k$. 
We have 
\[
|w \ra = \pi |v \ra = \pi X_k |v+e_k \ra = \pi X_k \pi^{-1} |u \ra = \ket{u+A_ku+e_k}.
\] 
Since $\alpha(u)=m>k$, and $\alpha(A_ku) >m >k$, we must have $\alpha(w) = \alpha(e_k+u+A_ku) = k = \alpha(v)$, which is a contradiction. 

We now build a polynomial representation (see~\Cref{def:polynomial_rep}) for $\pi^{-1}$. By~\Cref{ck-perm-poly}, every coordinate of $\pi^{-1}$ has degree at most $2$. 
Since 
$\pi^{-1}|0^n \ra = |0^n \ra$ and 
$\pi^{-1}|e_i \ra = |e_i \ra$ for all $i$, we have that the constant term of every coordinate is $0$ and the
linear term of the $i$-th coordinate is $a_i$ for all $i$.
Thus we can write $\pi^{-1}$ as 
\[
|a_1, \ldots, a_n \ra \mapsto |a_1+q_1, \ldots, a_n+q_n \ra,
\]
where each $q_k$ is a sum of some (possibly zero) monomials of the form $a_ia_j$ with $i<j$.

For any $i<j$, we have $\pi^{-1}|e_i+e_j \ra = |e_i+e_j+w_{ij} \ra$, where $w_{ij}$ has ones exactly at the positions $k$ for which $q_k$ contains the monomial $a_ia_j$. 
Let $v$ be such that $\pi|e_j+w_{ij} \ra = |v \ra$, we have
\begin{align*}
    \pi X_i \pi^{-1} |v \ra 
    &= \pi X_i |e_j+w_{ij} \ra = \pi |e_i+e_j+w_{ij} \ra = |e_i+e_j \ra \\
    &= \ket{v+A_iv+e_i},
\end{align*}
which means $v+A_iv = e_j$, and $j=\alpha(v+A_iv)=\alpha(v)$. 
Since $\pi|e_j+w_{ij} \ra = |v \ra$, we must also have $\alpha(e_j+w_{ij}) = \alpha(v)$. Thus $\alpha(e_j+w_{ij})=j$, which means $\alpha(w_{ij})>j$. 
In other words, any appearance of an $a_ia_j$ term can only be in a $q_k$ with $k>j$. Then \Cref{staircase-poly} implies that $\pi$ is in staircase form.

Thus, unraveling our without-loss-of-generality assumptions, we can write $\pi = \phi_1\mu\phi_2$ for Clifford permutations $\phi_1$ and $\phi_2$ and a staircase form permutation $\mu$. Finally, $\mu=\phi_1^{-1}\pi\phi_2^{-1}$ is in $\cC_3$ by \Cref{prop:cliff_hier}, since $\phi_1^{-1}, \phi_2^{-1} \in \cC_2$ and $\pi \in \cC_3$.
\end{proof}

\subsection{A consequence of \texorpdfstring{\Cref{thm:c3-staircase}}{Theorem 3.2}}

Below, in \Cref{cor:strengthen_bs}, we strengthen the result of Beigi and Shor \cite{beigi-shor}, who showed that any $\cC_3$ element is generalized semi-Clifford (restated in \Cref{c3-is-gsc}). Our corollary sharpens this by showing that the underlying permutation can be taken to be in staircase form.

\begin{lemma}
\label{pi-d-in-c3}
For any permutation gate $\pi$ and diagonal gate $d$, if $\pi d \in \cC_3$, then $\pi \in \cC_3$.
\end{lemma}

\begin{proof}
This is a special case of Corollary A.11.2 from \cite{anderson}.
\end{proof}

\begin{corollary}\label{cor:strengthen_bs}
For any $\psi \in \cC_3$, there exist $\phi_1, \phi_2 \in \cC_2$, a diagonal gate $d$, and a staircase form permutation $\pi \in \cC_3$ with $\psi = \phi_1 \pi d \phi_2$.
\end{corollary}

\begin{proof}
By \Cref{c3-is-gsc}, we can write $\psi = \phi_3 \pi_0 d_0 \phi_4$ for $\phi_3, \phi_4 \in \cC_3$, a permutation gate $\pi_0$, and a diagonal gate $d_0$. Now $\pi_0 d_0 = \phi_3^{-1} \psi \phi_4^{-1} \in \cC_3$, so $\pi_0 \in \cC_3$ by \Cref{pi-d-in-c3}. Then, by \Cref{thm:c3-staircase}, we can write $\pi_0 = \phi_5 \mu \phi_6$ for Clifford permutations $\phi_5, \phi_6$ and a staircase form permutation $\mu \in \cC_3$. Now $\phi_6 d_0 \phi_6^{-1}$ is diagonal as $\phi_6$ is a permutation and $d_0$ is diagonal. Let $\phi_1 = \phi_3 \phi_5$, $\pi = \mu$, $d = \phi_6 d_0 \phi_6^{-1}$, $\phi_2 = \phi_6 \phi_4$. Then $\phi_1, \phi_2 \in \cC_2$, $\pi$ is a staircase form permutation in $\cC_3$, and $d$ is diagonal, and 
\begin{align*}
    \phi_1 \pi d \phi_2 &= \phi_3 \phi_5 \mu \phi_6 d_0 \phi_6^{-1} \phi_6 \phi_4 = \phi_3 \phi_5 \mu \phi_6 d_0 \phi_4 = \phi_3 \pi_0 d_0 \phi_4 = \psi. \qedhere
\end{align*}
\end{proof}

\section{Descending multiplications}\label{sec-descending-mult}

As we proved that every permutation in $\cC_3$ can be written in staircase form, in this section we show sufficient conditions for a staircase form permutation to be in $\cC_3$.
Previously we have been working with $\mathbb{F}_2^n$ as a vector space without any multiplicative structure. 
Our next definition captures bilinear products defined over $\mathbb{F}_2^n$ which encode staircase form Toffoli circuits in $\cC_3$.

\begin{definition}[Descending multiplications]\label{def:descending-mult}
    A map $\mathbb{F}_2^n \times \mathbb{F}_2^n \rightarrow \mathbb{F}_2^n$, denoted by juxtaposition, is called a \emph{descending multiplication} if
    \begin{itemize}
        \item it is linear in each coordinate (distributive), associative, and commutative,
        \item for all $i \in [n]$, we have $e_ie_i = e_i^2=0$, and
        \item for all $i < j \in [n]$, we have $e_ie_j$ is in the span of $\{e_k: k> j\}$.
    \end{itemize}
\end{definition}

Observe that, for any descending multiplication, we have $v^2=0$ for all $v$. 
Henceforth, for any permutation gate $\pi$, we will also interpret $\pi$ as a permutation of $\mathbb{F}_2^n$, so that whenever $\pi|v \ra = |w \ra$, we can write $\pi(v)=w$.
We will also write $0$ and $0^n$ interchangeably. 

We now state the main theorem of this section. 

\begin{theorem}[A bijection between descending multiplications and $\cC_3$ permutations in staircase form]\label{thm:perm-is-mult}
There is a one-to-one correspondence of descending multiplications to $\cC_3$ permutations in staircase form, which we describe as follows. 
For each descending multiplication, the corresponding $\cC_3$ permutation $\pi$ is given by
\begin{equation}\label{mult-to-pi}
    \forall S\subseteq [n], \quad \pi \ket{\sum_{i \in S} e_i} = \ket{\sum_{T \subseteq S, T \neq \varnothing} \prod_{i \in T} e_i}. 
\end{equation}
Each permutation $\pi \in \cC_3$ in staircase form induces a multiplication operation where each $e_ie_j$ is given by
\begin{equation}\label{pi-to-mult}
    \pi \ket{e_i+e_j} = \ket{e_i+e_j+e_ie_j}.
\end{equation}
The multiplication is then extended linearly.
\end{theorem}

We break the proof of this theorem into several propositions. 

\begin{restatable}{proposition}{DMmulttopiwelldefined}
\label{mult-to-pi-well-defined}
For any descending multiplication, the resulting $\pi$ from \Cref{mult-to-pi} is indeed a staircase form permutation in $\cC_3$.
\end{restatable}

\begin{restatable}{proposition}{DMpitomultwelldefined}
\label{pi-to-mult-well-defined}
For any staircase form permutation $\pi$ in $\cC_3$, the resulting operation from \Cref{pi-to-mult} is indeed a descending multiplication.
\end{restatable}

\begin{restatable}{proposition}{DMmulttopitomult}
\label{mult-to-pi-to-mult}
Given a descending multiplication, applying the procedure in \Cref{mult-to-pi} to get a permutation $\pi$, then applying the procedure in \Cref{pi-to-mult} to that permutation, yields the original multiplication.
\end{restatable}

\begin{restatable}{proposition}{DMpitomulttopi}
\label{pi-to-mult-to-pi}
Given a staircase form permutation $\pi$ in $\cC_3$, applying the procedure in \Cref{pi-to-mult} to get a descending multiplication, then applying the procedure in \Cref{mult-to-pi} to that multiplication, yields the original $\pi$.
\end{restatable}

\begin{proof}[Proof of \Cref{thm:perm-is-mult}]
Combining \Cref{mult-to-pi-well-defined,pi-to-mult-well-defined,mult-to-pi-to-mult,pi-to-mult-to-pi} yields the theorem.
\end{proof}

For the rest of this section, we will include two technical lemmas in \Cref{sec:DM_lemmas} that are helpful for proving \Cref{mult-to-pi-well-defined,pi-to-mult-well-defined,mult-to-pi-to-mult,pi-to-mult-to-pi}, and then we will prove these propositions in \Cref{sec:DM_props_proofs}. 

\subsection{Helper lemmas}\label{sec:DM_lemmas}

\begin{lemma}\label{staircase-conditions}
Let $\pi \in \cC_3$ be a permutation. Then $\pi$ is in staircase form if and only if the following conditions hold:

\begin{itemize}
\item $\pi|0 \ra = |0 \ra$,

\item $\pi|e_i \ra = |e_i \ra$ for all $i\in [n]$, and

\item for any vector $v$ with at least two $1$s, the indices of the first two $1$s in $v$ and $\pi(v)$ are the same.
\end{itemize}
\end{lemma}

\begin{proof}
We use the equivalent characterization of staircase form permutations in \Cref{staircase-poly}.

For the ``if'' direction, consider the polynomial representation $(b_1, \dots, b_n)$ of $\pi^{-1}$. We know from \Cref{ck-perm-poly} that $b_i$ has degree at most $2$ for all $i$. The condition $\pi^{-1}(0)=0$ yields that the constant term of $b_i$ is $0$ for all $i$. 
Moreover, for all $i\in [n]$, the condition $\pi^{-1}(e_i)=e_i$ yields that the linear term of $b_i$ is $a_i$. 
Thus we can write $b_k=a_k+q_k$ where $q_k$ is a sum of monomials of the form $a_ia_j$ with $i<j$.
To show that $\pi$ is in staircase form using \Cref{staircase-poly}, it remains to show that each monomial $a_ia_j$ in $q_k$ with $i<j$ must also satisfy $j<k$. 

Now for any $i<j$, $\pi^{-1}(e_i+e_j)$ cannot be $0$ or any $e_k$ (since $\pi$ is a permutation).
Therefore, $\pi^{-1}(e_i+e_j)$ considered as a vector has at least two $1$s.
Invoking the third condition tells us that the indices of the first two $1$s in $\pi^{-1}(e_i+e_j)$ and $\pi(\pi^{-1}(e_i+e_j)) = e_i+e_j$ are the same. 
So we can write $\pi^{-1}(e_i+e_j) = e_i+e_j+x_{ij}$ where $x_{ij}$ is a (possibly empty) sum of terms of the form $e_{\ell}$ for $\ell > j$. 
Recall that $b_k = a_k+q_k$ is the $k$-th coordinate of the polynomial representation of $\pi^{-1}$. 
So if $q_k$ contains an $a_ia_j$ term, then $x_{ij}$ must contain an $e_k$ term, which implies that $j<k$.

For the ``only if'' direction, suppose $\pi$ is staircase form. We know from \Cref{staircase-poly} that we can write the polynomial representation of $\pi^{-1}$ as $(a_1+q_1, \dots, a_n+q_n)$ where $q_k$ is a (possibly empty) sum of terms of the form $a_ia_j$ with $i<j<k$. 
This implies that that $\pi(0)=0$ and $\pi(e_i)=e_i$ for all $i$.
To prove the third condition, consider $v$ containing at least two $1$s. 
Let $w=\pi(v)$; then $w\ne 0$ and $w\ne e_i$ for all $i$, so it contains at least two $1$s. 
Let $I<J$ be the positions of the first two $1$s in $w$.
Suppose that $w$ and $\pi^{-1}(w)=v$ differ on the $k$-th coordinate, i.e.\ $0\neq a_k(w) + b_k(w) = q_k(w)$. 
This means that there is
a term of $q_k$, say $a_ia_j$ with $i<j<k$, is nonzero when evaluated at $w$. 
In other words $a_i=a_j=1$ in $w$, which means $j \geq J$ and $k>J$. 
Thus $w$ and $v$ agree on the first $J$ coordinates; in particular, they agree on the positions of the first two $1$s.
\end{proof}

\begin{lemma}
For any descending multiplication, 
the resulting $\pi$ from \Cref{mult-to-pi} satisfies that
\begin{equation}\label{pi-v-w}
\pi(v+w)=\pi(v)+\pi(w)+\pi(v)\pi(w), \quad \text{for all }v,w\in \mathbb{F}_2^n. 
\end{equation}
\end{lemma}
\begin{proof}
Let us write $v = \sum_{i \in V} e_i$ and $w = \sum_{i \in W} e_i$ for some subsets $V,W\subseteq [n]$. 
Let $A = V \cap W$, $B = V \backslash W$, and $C = W \backslash V$. 
Let $a = \sum_{i \in A} e_i, b = \sum_{i \in B} e_i, c = \sum_{i \in C} e_i$. 
Since $A \cup B = V$ and $A \cap B = \varnothing$, it follows from \Cref{mult-to-pi} that
\begin{align*} 
\pi(a) + \pi(b) + \pi(a)\pi(b) &= \sum_{T \subseteq A, T \neq \varnothing} \prod_{i \in T} e_i + \sum_{T \subseteq B, T \neq \varnothing} \prod_{i \in T} e_i + \left(\sum_{T \subseteq A, T \neq \varnothing} \prod_{i \in T} e_i\right)\left(\sum_{T \subseteq B, T \neq \varnothing} \prod_{i \in T} e_i\right) \\ 
&= \sum_{T \subseteq A \cup B, T \neq \varnothing} \prod_{i \in T} e_i = \pi(v). 
\end{align*}

We can similarly show that $\pi(a)+\pi(c)+\pi(a)\pi(c) = \pi(w)$ and $\pi(b)+\pi(c)+\pi(b)\pi(c) = \pi(v+w)$. 
Further, since $\pi(c)\pi(c)=0$, we have \begin{align*} 
\pi(v) + \pi(w) + \pi(v)\pi(w)
&= \pi(a)+\pi(b)+\pi(a)\pi(b) + \pi(a)+\pi(c)+\pi(a)\pi(c) \\ 
&\quad +(\pi(a)+\pi(b)+\pi(a)\pi(b))(\pi(a)+\pi(c)+\pi(a)\pi(c)) \\ 
&= \pi(b)+\pi(a)\pi(b)+\pi(c)+\pi(a)\pi(c) \\
&\quad + (\pi(a)\pi(b)+\pi(a)\pi(c)+\pi(b)\pi(c)) \\ 
&= \pi(b)+\pi(c)+\pi(b)\pi(c) \\
&= \pi(v+w). \qedhere
\end{align*}
\end{proof}

\subsection{Proofs of \texorpdfstring{\Cref{mult-to-pi-well-defined,pi-to-mult-well-defined,mult-to-pi-to-mult,pi-to-mult-to-pi}}{Propositions 4.3 to 4.6}}\label{sec:DM_props_proofs}

We now restate and prove the four key propositions individually.

\DMmulttopiwelldefined*
\begin{proof}
We first note that \Cref{mult-to-pi} implies the three conditions in \Cref{staircase-conditions}. 
In particular, the fact that $e_ie_j$ is a descending multiplication implies the third condition. 
Therefore, it suffices for us to show that \(\pi\) is a permutation gate in $\cC_3$.

Since $\pi(v)$ is nonzero for $v \neq 0$, $\pi$ is injective: if $\pi(v)=\pi(w)$, then \Cref{pi-v-w} implies that $\pi(v+w)=0$, which means $v=w$. Thus $\pi$ is a valid permutation of $\mathbb{F}_2^n$.

To show $\pi\in \cC_3$, we first prove that $\pi X_i \pi^{-1} \in \cC_2$ for each $i$.
For any $v$ and any index $i$, if we let $\pi^{-1}(v) = w$, then 
\[
\pi X_i \pi^{-1}(v) = \pi(e_i+w) = \pi(e_i)+\pi(w)+\pi(e_i)\pi(w) = e_i+v+e_iv.
\] 
The map $v \mapsto e_i+v+e_iv$ is invertible since it is its own inverse; thus $\pi X_i \pi^{-1} \in \cC_2$ by \Cref{clifford-perm}.

Next, to prove that $\pi Z_i \pi^{-1} \in \cC_2$ for all $i$, by \Cref{eq:pi_conjugate_on_Z,diagonal-in-ch}, it suffices to show that the $i$-th coordinate in the polynomial representation of $\pi^{-1}$ has degree at most $2$. 

For each pair of indices $i<j$, define $v_{ij}$ to be such that $\pi(e_i+e_j+v_{ij})=e_i+e_j$. 
We will show that for any set $S$ of indices, 
\begin{equation} \label{pi-inverse-formula}
\pi \left(\sum_{i \in S} e_i + \sum_{i,j \in S; i<j} v_{ij} \right) = \sum_{i \in S} e_i.
\end{equation}
This implies that every coordinate of $\pi^{-1}$ has degree at most $2$, since each vector $v_{ij}$ appear in the sum when both $i,j$ are in $S$. 

For the rest of the proof, we will prove \Cref{pi-inverse-formula} by induction on $|S|$. 
The base case $|S| \leq 2$ is clear. 
For the inductive step, 
let $i$ be the smallest element of $S$, and let $T = S \backslash \{i\}$. Let $w = \sum_{j \in T} e_j$, and let $x=e_i + \sum_{j \in T} v_{ij}$. 
By the inductive hypothesis, $\pi(w + \sum_{j,k \in T; j<k} v_{jk}) = w$. 
Thus 
\begin{align}\label{eq:induction}
    \pi \left(\sum_{j \in S} e_j + \sum_{j,k \in S; j<k} v_{jk} \right) &= \pi \left( e_i + \sum_{j\in T} e_j + \sum_{j,k\in T; j<k} v_{jk} + \sum_{j\in T}v_{ij}\right) \nonumber \\
    &= \pi \left(\Big(w + \sum_{j,k \in T; j<k}v_{jk}\Big) + x\right) \nonumber\\
    &= w + \pi(x) + w\pi(x) .
\end{align}
where the last equality follows from \Cref{pi-v-w} and the induction hypothesis. 
Observe that 
\begin{equation}\label{eq:pi(vij)}
\begin{split}
    \pi(v_{ij}) &= \pi((e_i+e_j+v_{ij})+(e_i+e_j)) \\
    &= (e_i+e_j)(e_i+e_j+e_ie_j)+(e_i+e_j)+(e_i+e_j+e_ie_j) \\
    &= e_ie_j.
\end{split}
\end{equation}
We can show that $\pi(x) = e_i + \sum_{j\in T}e_ie_j$ via a simple induction on $|T|$, as well as using \Cref{pi-v-w} and the facts that $e_i^2=0$, $\pi(e_i) = e_i$, and $\pi(v_{ij})=e_ie_j$ repeatedly. 
Now continuing \Cref{eq:induction}:
\begin{align*}
    \pi \left(\sum_{j \in S} e_j + \sum_{j,k \in S; j<k} v_{jk} \right) &= w + (e_i+e_iw) + w(e_i+e_iw) = w+e_i = \sum_{j \in S} e_j. 
\end{align*}
This completes the inductive step, so \Cref{pi-inverse-formula} holds. 
\end{proof}

\DMpitomultwelldefined*
\begin{proof}
The distributive property holds by definition. The commutative property clearly holds, since $e_ie_j=e_je_i$. We have $e_i^2=0$ since $\pi|0 \ra = |0 \ra$. The fact that $e_ie_j$ is in the span of $\{e_k : k>j\}$ for $i<j$ follows from \Cref{staircase-conditions}.
It remains for us to prove associativity.

For each $i$, we have that $\pi X_i \pi^{-1}$ is a Clifford permutation, so by \Cref{clifford-perm}, there is a matrix $A_i$ and a vector $b_i$ so that $\pi X_i \pi^{-1}|v \ra = |v+A_i v+ b_i \ra$ for all $v$. Setting $v=0$ yields that $b_i=e_i$. Then setting $v=e_j$ yields that $\pi|e_i+e_j \ra = |e_i+e_j+A_ie_j \ra$, so $A_ie_j = e_ie_j$. Since we defined the multiplication to be linear in each coordinate, this implies $A_iv = e_iv$ for all $v$.

Now, since $X_i$ and $X_k$ commute, so do $\pi X_i \pi^{-1}$ and $\pi X_k \pi^{-1}$.
Therefore the maps $v \mapsto v+A_i v + e_i$ and $v \mapsto v + A_kv + e_k$ commute, which means $A_i$ and $A_k$ commute. Thus $A_i(A_ke_j) = A_k(A_ie_j)$ for all $i,j,k$, so $e_i(e_je_k) = e_i(e_ke_j) = e_k(e_ie_j) = (e_ie_j)e_k$, by commutativity of the multiplication. This yields the desired associativity.
\end{proof}

\DMmulttopitomult*
\begin{proof}
For $i\neq j$, by setting $S = \{i,j\}$ in \Cref{mult-to-pi}, we see that the new multiplication has the same value of $e_ie_j$ as the original multiplication; also, setting $j=i$ in \Cref{pi-to-mult} gives that the new multiplication has $e_i^2=0$. Thus the original and new multiplications coincide on the value of $e_ie_j$ for all $i,j$, and both are linear in each input. This means they are the same multiplication.
\end{proof}

\DMpitomulttopi*

\begin{proof}
Let $\pi$ denote the original permutation.
For any set $S$ of indices, we have 
\[ 
\pi\ket{\sum_{i \in S} e_i} = \pi \left(\prod_{i \in S} X_i \right) \pi^{-1}|0\ra = \prod_{i \in S} (\pi X_i \pi^{-1})|0 \ra.
\]
From the proof of \Cref{pi-to-mult-well-defined}, $\pi X_i \pi^{-1} |v \ra = |v+e_iv+e_i \ra$ for all $v,i$. 
It can be easily verified that applying the maps $v \mapsto v+e_iv+e_i$ sequentially for all $i \in S$, starting with $v=0$, yields $\sum_{T \subseteq S, T \neq \varnothing} \prod_{i \in T} e_i$ as desired.
\end{proof}

\section{A family of non-semi-Clifford Permutations}\label{sec:family_gates}

Utilizing our characterization of $\cC_3$ permutations as descending multiplications, we construct an infinite family of permutation gates $\{U_k\}_{k\geq 3}$ in \Cref{def:family_gates} and prove in \Cref{thm:U_k_non_semi_cliff} that every $U_k$ is non-semi-Clifford in $\cC_3$. 
This family of gates disproves both of Anderson's conjectures (restated in \Cref{conj:c3-tof} and \Cref{conj:perm-inverse} for convenience).  
We study the smallest case of $k=3$ in \Cref{subsec:family_example_k3}, and show that this $7$-qubit permutation $U_3$ is in fact conjugate to the Gottesman--Mochon example by a Clifford operator.

\begin{definition}[A family of permutation gates in $\cC_3$ from descending multiplications]\label{def:family_gates}
    For each integer $k \geq 3$, let $n=2^k-1$.
    We label a basis of $\mathbb{F}_2^n$ as $e_S$ for nonempty subsets $S \subseteq [n]$. 
    Define a multiplication by setting $e_Se_T = e_{S \sqcup T}$ if $S \cap T = \varnothing$, and $e_Se_T = 0$ if $S \cap T \neq \varnothing$, and extending linearly; it is easy to check that this yields a descending multiplication. 
    Let $U_k$ be the staircase form $\cC_3$ permutation corresponding to this descending multiplication, as in \Cref{thm:perm-is-mult}.
\end{definition}

\begin{proposition}\label{prop:toffoli-form-of-uk}
We can express $U_k$ as a product of Toffoli gates in staircase form as follows: for each pair $S,T$ of nonempty disjoint subsets of $[n]$ with $S<T$, apply the gate $\tof_{S, T, S \sqcup T}$; specifically, apply these gates in nondecreasing order of target.
\end{proposition}

\begin{remark}
From the perspective of labeling gates with integers instead of sets, \Cref{prop:toffoli-form-of-uk} states that we can express $U_k$ as a product of Toffoli gates in staircase form as follows: for each pair of indices $i<j$ that do not have any $1$s in the same place as each other in binary, apply $\tof_{i,j,i+j}$; specifically, apply these gates in nondecreasing order of target.
\end{remark}

The main feature of these permutations is captured by the following theorem.

\begin{theorem}[Inverses of $\cC_3$ permutations may lie outside $\cC_k$]\label{thm:U_k_non_semi_cliff}
    For any integer $k\geq 3$, we have $U_k\in \cC_3$ but $U_k^{-1} \notin \cC_k$. 
    Thus $U_k$ is not semi-Clifford.
\end{theorem}

\subsection{Proof of \texorpdfstring{\Cref{prop:toffoli-form-of-uk}}{Proposition 5.2}}

\begin{proposition}\label{prop:toffoli-form-of-pi}
Given a staircase form permutation $\pi$ in $\cC_3$ and $i < j$, let $v_{ij}$ be such that $\pi^{-1}(e_i+e_j) = e_i+e_j+v_{ij}$. Then for all $i<j<k$, $\tof_{i,j,k}$ appears in the staircase form of $\pi$ if and only if there is an $e_k$ term in $v_{ij}$.
\end{proposition}
\begin{proof}
Note that $v_{ij}$ is the vector of the positions containing a $a_ia_j$ monomial in the polynomial form of $\pi^{-1}$, so this follows directly from \Cref{staircase-poly}.
\end{proof}

\begin{proof}[Proof of \Cref{prop:toffoli-form-of-uk}]
Take $U_k=\pi$ in \Cref{prop:toffoli-form-of-pi}. We know from the proof of \Cref{mult-to-pi-well-defined} that $v_{ij}=\pi^{-1}(e_ie_j)$ for all $i<j$. In the notation of indexing qubits as sets, this becomes $v_{ST} = \pi^{-1}(e_Se_T)$ for $S<T$. Now note that $e_Se_T$ is $e_{S \cup T}$ or $0$, and in either case $\pi^{-1}(e_Se_T) = e_Se_T$. Thus $v_{ST}=e_Se_T$ in any case, so \Cref{prop:toffoli-form-of-pi} implies the desired, by the definition of $e_Se_T$.
\end{proof}

\subsection{Proof of \texorpdfstring{\Cref{thm:U_k_non_semi_cliff}}{Theorem 5.4}}

\begin{proposition} \label{prod-ei-and-ai}
Given a descending multiplication and its corresponding $\cC_3$ permutation $\pi$, 
for any nonempty set $S\subseteq [n]$, the value of $\Pi_{i \in S} e_i$ is exactly the vector corresponding to the positions at which an $\Pi_{i \in S} a_i$ term appears in the polynomial representation of $\pi$.
\end{proposition}

\begin{proof}
For any $S$, let $p_S$ be the vector corresponding to the positions at which an $\Pi_{i \in S} a_i$ term appears in the polynomial representation of $\pi$. Then we can see, for any set $S$, that $\pi(\sum_{i\in S} e_i) = \sum_{T \subseteq S} p_T$. 
From \Cref{mult-to-pi}, we have that $\sum_{T \subseteq S; T \neq \varnothing} p_T = \sum_{T \subseteq S, T \neq \varnothing} \prod_{i \in T} e_i$ for any nonempty set $S$. It then easily follows by strong induction on $|S|$ that $p_S = \prod_{i \in S} e_i$ for all nonempty $S$.
\end{proof}

We prove \Cref{thm:U_k_non_semi_cliff} through the polynomial representations of $U_k$ and $U_k^{-1}$.

\begin{lemma}\label{lem:U_k_poly_rep}
    For each integer $k\geq 3$, let us denote the polynomial representations (\Cref{def:polynomial_rep}) of $U_k$ and $U_k^{-1}$ respectively as $(\pi_S)_{S\subseteq [k], S\neq \varnothing}$ and $(\pi_S')_{S\subseteq [k], S\neq \varnothing}$. 
    Then
    \begin{equation}\label{eq:poly_rep_U_k}
        \pi_S(a) = \sum_{m=1}^{|S|} \ \sum_{\substack{\bigsqcup_{i=1}^m T_i = S, \\ T_i\neq \varnothing}} a_{T_1}a_{T_2}\dots a_{T_m}, 
    \end{equation}
    where the sum is over unordered non-empty subsets $T_1, \ldots, T_m$, in other words, all partitions of $S$, 
    and 
    \begin{equation*}
        \pi'_S(a) = a_S + \sum_{T_i\neq \varnothing, T_1 \sqcup T_2 = S} a_{T_1}a_{T_2}, 
    \end{equation*}
    where the sum is over unordered pairs $T_1, T_2$. 
\end{lemma}

\begin{proof}
The polynomial form of $\pi^{-1}$ follows directly from \Cref{prop:toffoli-form-of-uk} and \Cref{staircase-poly}. For the polynomial form of $\pi$, observe that, for any $T_1, \dots, T_m$, we have $e_{T_1}\dots e_{T_m}$ is $e_{T_1 \sqcup \dots \sqcup T_m}$ if the $T_i$ are pairwise disjoint, and $0$ if not, so the desired follows from \Cref{prod-ei-and-ai}. 
\end{proof}

\begin{proof}[Proof of~\Cref{thm:U_k_non_semi_cliff}]
We know $U_k \in \cC_3$ by definition. Also, by \Cref{lem:U_k_poly_rep}, 
we know that $\pi_{[k]}$ is a polynomial of degree $k$ because it contains the monomial $a_{\{1\}}a_{\{2\}}\dots a_{\{k\}}$. 
This implies that $U_k^{-1} \notin \cC_k$ using \Cref{ck-perm-poly}. In particular, $U_k^{-1} \notin \cC_3$, so $U_k$ is not semi-Clifford by \Cref{lem:respect-sc}. 
\end{proof}

\subsection{Example: \texorpdfstring{$k=3$}{k=3}}\label{subsec:family_example_k3}

Let us examine the case of $k=3$, the simplest gate in the family. 
$U_3$ is a $7$-qubit permutation gate with $6$ Toffoli gates in staircase form (see circuit in~\Cref{fig:seven_perm}): 
\begin{equation*}
    U_3 = \tof_{1,6,7}\tof_{2,5,7}\tof_{3,4,7}\tof_{2,4,6}\tof_{1,4,5}\tof_{1,2,3}.
\end{equation*}

\begin{figure}[!ht]
    \centering
    \begin{equation*}
        \begin{quantikz}[slice style=blue] 
        \lstick{$a_1$}&\ctrl{2}&\ctrl{4}&&&&\ctrl{6}&\rstick{$a_1$}\\
        \lstick{$a_2$}&\control{}&&\ctrl{4}&&\ctrl{5}&&\rstick{$a_2$}\\
        \lstick{$a_3$}&\targ{}&&&\ctrl{4}&&&\rstick{$a_3+a_1a_2$}\\
        \lstick{$a_4$}&&\control{}&\control{}&\control{}&&&\rstick{$a_4$}\\
        \lstick{$a_5$}&&\targ{}&&&\control{}&&\rstick{$a_5+a_1a_4$}\\
        \lstick{$a_6$}&&&\targ{}&&&\control{}&\rstick{$a_6+a_2a_4$}\\
        \lstick{$a_7$}&&&&\targ{}&\targ{}&\targ{}&\rstick{$a_7+a_1a_6+a_2a_5+a_3a_4+a_1a_2a_4$}
        \end{quantikz}
    \end{equation*}
    \caption{The non-semi-Clifford permutation gate $U_3\in \cC_3$.}
    \label{fig:seven_perm}
\end{figure}
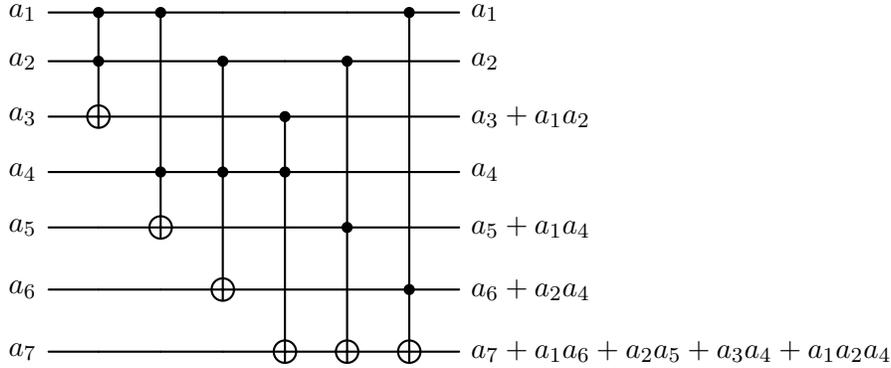

Recall the Gottesman--Mochon $7$-qubit gate (see~\Cref{fig:circ_gottesman_Mochon} for the circuit)
\begin{align*}
    \cC_3 \ni G &= \cswap_{7,1,6}\cswap_{7,2,5}\cswap_{7,4,3} \cdot \ccz_{1,2,3}\ccz_{1,4,5}\ccz_{2,4,6}\ccz_{3,5,6}, 
\end{align*}
which is not semi-Clifford, as in the proof of~\Cref{lem:gottesman-mochon}. 
Let us define a $7$-qubit Clifford gate $F$: 
\begin{align*}
    F &= H_3H_5H_6\cnot_{6,1}\cnot_{5,2}\cnot_{3,4}H_7.
\end{align*}

\begin{proposition} \label{prop:u3}
$U_3$ is a non-semi-Clifford permutation in $\cC_3$ on $7$ qubits and $FGF^{-1}=U_3$.
\end{proposition}

\begin{proof}
The fact that $U_3$ is a non-semi-Clifford permutation on $7$ qubits is a special case of \Cref{thm:U_k_non_semi_cliff}. Checking that $FGF^{-1}=U_3$ is a direct computation. 
\end{proof}

By~\Cref{sc-lowering}, this implies that for all $n \geq 7$, $\cC_3$ contains a non-semi-Clifford permutation.

\subsection{Lower bound on the number of qubits required for a \texorpdfstring{$\cC_3$}{C₃} permutation}

From~\Cref{thm:U_k_non_semi_cliff}, we know that $U_k$ is a $\cC_3$ permutation on $2^k-1$ qubits that contains a degree-$k$ monomial in its polynomial representation.
Our next result shows that any such permutation in $\cC_3$ must be supported on at least $2^k-1$ qubits.

\begin{restatable}[Lower bound on the number of qubits required for a $\cC_3$ permutation]{theorem}{thmlowerbound}
\label{thm:c3-degree-qubit}
If $\pi$ is a $\cC_3$ permutation such that there is a degree-$k$ monomial somewhere in the polynomial representation of $\pi$, then $n \geq 2^k-1$; this bound is sharp for all $k \geq 3$.
\end{restatable}

We prove \Cref{thm:c3-degree-qubit} by taking advantage of the one-to-one correspondence between descending multiplications (\Cref{prod-ei-nonzero}) and $\cC_3$ permutations in staircase form (\Cref{staircase-degree}). 

\begin{proposition} \label{prod-ei-nonzero}
Given any descending multiplication, if $\Pi_{i \in S} e_i$ is nonzero for some $k$-element set $S$, then $n \geq 2^k-1$.
\end{proposition}
\begin{proof}
Suppose $\Pi_{i \in S} e_i$ is nonzero. For any nonempty subset $T$ of $S$, let $p_T = \Pi_{i \in T} e_i$. We shall show that the vectors $p_T$, over all nonempty subsets $T$ of $S$, are linearly independent.

Suppose for contradiction they are linearly dependent. Take a family $F$ of nonempty subsets of $S$ such that $\sum_{T \in F} p_T = 0$. Take a minimal element $U$ of $F$ (i.e.\ such that no proper subset of $U$ is in $F$). Let $V = S \backslash U$. 
Note that for any $T \in F$ with $T \neq U$, we have $T \not\subseteq U$, which means $T\cap V\ne \varnothing$. From \Cref{def:descending-mult}, we see that $p_Tp_V = 0$. 
Therefore, 
\[
\sum_{T \in F} p_Tp_V = p_Up_V = p_S.
\]
On the other hand, $\sum_{T \in F} p_T = 0$ implies $(\sum_{T \in F} p_T)p_V = 0$. 
Thus $p_S = 0$, which is a contradiction. We conclude that the vectors $p_T$ must be linearly independent. 
Since there are $2^k-1$ such vectors in the vector space $\mathbb{F}_2^n$, we must have $n \geq 2^k-1$.
\end{proof}

\begin{proposition}\label{staircase-degree}
If $\pi$ is a staircase form $\cC_3$ permutation such that there is a degree-$k$ monomial somewhere in the polynomial form of $\pi$, then $n \geq 2^k-1$. 
\end{proposition}
\begin{proof}
Suppose $\Pi_{i \in S} a_i$ appears somewhere in the polynomial form of $\pi$, for some $k$-element set $S$. Then, for the descending multiplication corresponding to $\pi$, we have $\Pi_{i \in S} e_i$ is nonzero, by \Cref{prod-ei-and-ai}; then \Cref{prod-ei-nonzero} yields the desired.
\end{proof}

We are now ready to prove \Cref{thm:c3-degree-qubit}. 

\begin{proof}[Proof of \Cref{thm:c3-degree-qubit}]
Assume without loss of generality that $k \geq 2$. By \Cref{thm:c3-staircase}, we can write $\pi = \phi_1 \mu \phi_2$, for Clifford permutations $\phi_1$ and $\phi_2$ and staircase form $\mu \in \cC_3$. Now all terms in the polynomial forms of $\phi_1$ and $\phi_2$ are degree at most $1$; then if all terms in the polynomial form of $\mu$ have degree less than $k$, then all terms in the polynomial form of $\phi_1 \mu \phi_2 = \pi$ have degree less than $k$, a contradiction. Thus there exists some term in the polynomial form of $\mu$ with degree $d \geq k$. Then \Cref{staircase-degree} applied to $\mu$ yields that $n \geq 2^d-1 \geq 2^k-1$, as desired. The example of $U_k$ shows that the bound is sharp.
\end{proof}

\subsection{Rejection of Anderson's conjectures}
\label{subsec:anderson-conj}

We now disprove Anderson's two conjectures. 

\begin{lemma} \label{lem:mismatch-free-equals-commute}
Two multi-controlled NOT gates commute if and only if they have no mismatch (that is, there is no qubit that is used as a target in one and a control in the other).
\end{lemma}

\begin{proof}
The ``if'' direction is clear. Let us prove the ``only if'' direction. 
Assume for the sake of contradiction that they have mismatch. Without loss of generality, let the gates be $A$, with qubit $1$ as a control and qubit $2$ as target, and $B$, with qubit $1$ as target. Then $AB|1^n \ra = A|01^{n-1} \ra = |01^{n-1} \ra$, while $BA|1^n \ra = B|101^{n-2} \ra$, which is either $|101^{n-2} \ra$ or $|001^{n-2} \ra$; in either case $BA|1^n \ra \neq AB|1^n \ra$, so they do not commute.
\end{proof}

\begin{proposition}
\Cref{conj:c3-tof} and~\Cref{conj:perm-inverse} are false.
\end{proposition}

\begin{proof}
\Cref{conj:perm-inverse} is false by \Cref{thm:U_k_non_semi_cliff}.
For \Cref{conj:c3-tof}, suppose it holds.
Then by~\Cref{lem:mismatch-free-equals-commute}, every permutation in $\cC_3$ is a mismatch-free product of Toffoli gates, up to Clifford permutations on the left and right. 
This implies that every permutation in $\cC_3$ is semi-Clifford by~\Cref{mismatch-free-is-sc}. 
This is a contradiction, as we know $U_3$ is a non-semi-Clifford permutation in $\cC_3$ for $n=7$. 
\end{proof}

\section*{Acknowledgments}
The work was conducted as a part of the 2024 Summer Program for Undergraduate Research (SPUR) and 2024-2025 Undergraduate Research Opportunities Program (UROP) at MIT. 
We thank Jonathan Bloom, Isaac Chuang, David Jerison, and Peter Shor for their mentorship. 
X.~Tan would like to thank Robert Calderbank for introducing the problem of characterizing the third-level Clifford hierarchy to her in 2022 and many insightful discussions afterwards. 
We thank Jonas Anderson, Jeongwan Haah, Aram Harrow, Greg Kahanamoku-Meyer, Andrey Khesin and Anirudh Krishna for helpful discussions. 

Z.~He is supported by National Science Foundation Graduate Research Fellowship under Grant No.~2141064. 
X.~Tan is supported by the U.S. Department of Energy, Office of Science, National Quantum Information Science Research Centers, Co-design Center for Quantum Advantage (C2QA) under contract number DE-SC0012704.

\bibliographystyle{alpha}
\bibliography{references}

\appendix

\section{The smallest non-semi-Clifford permutation} \label{sec-seven-is-best}

In this appendix, we show that all permutations in $\cC_3$ supported on at most $6$ qubits are semi-Clifford, while for all $n\ge 7$ there is a non-semi-Clifford permutation gate in $\cC_3$.

Suppose an $n$-qubit unitary $U\in \cC_k$ is not semi-Clifford. 
It is trivial to see the $(n+1)$-qubit unitary $U\otimes I_2$ is in $\cC_k$, but it is not completely trivial to conclude that $U\otimes I_2$ is also not semi-Clifford. 
It is sometimes glossed over in the literature (e.g.\ the proof of~\cite[Theorem 3]{zeng2008semi-clifford}). 
We here provide a more careful treatment of this fact. 

\begin{fact} \label{structure-of-max-ab-subgp}
Any maximal abelian subgroup of $\cP_n$ is isomorphic to $(\mathbb{Z}/2\mathbb{Z})^n$ up to phase.
\end{fact}

\begin{lemma} \label{new-max-ab-subgp}
Let $A$ be a maximal abelian subgroup of $\cP_n$, and let $B$ be a (not necessarily maximal) abelian subgroup of $\cP_n$. Then there exists a maximal abelian subgroup $A'$ of $\cP_n$ such that $B \subseteq A' \subseteq \la A, B \ra$.
\end{lemma}

\begin{proof}
We first consider the case where $B$ is generated by a single operator $b$ (up to phase). 
Let $\{a_1, \ldots, a_n\}$ be the generators of $A$ up to phase. 
Without loss of generality, suppose $a_1, \ldots, a_k$ are the generators of $A$ which anti-commute with $b$. Consider sequential pairwise products of the form $a_1a_2, a_2a_3, a_3a_4, \ldots, a_{k-1}a_k$, and let $A'$ be the group generated by $\{b, a_1a_2, \ldots, a_{k-1}a_k, a_{k+1}, \ldots, a_n\}$ (up to $\{\pm 1, \pm i\}$ phase). We see that $A'$ is a maximal abelian subgroup of $\cP_n$ and $B\subseteq A'\subseteq \langle A, B\rangle$.

In the case where $B$ is generated by operators $b_1, \ldots, b_k$, we iteratively update $A'$ for every generator of $B$ with the above procedure. Note that the update procedure can be seen to preserve all elements of $A$ that commute with $b$. Thus, at every update, we keep all generators of $B$ that were already added (as $B$ is abelian); thus we obtain the desired subgroup.
\end{proof}

\begin{lemma} \label{sc-lowering}
Suppose $U$ is an $n$-qubit non-semi-Clifford gate. Then $U\otimes I_{2^m}$ is also not semi-Clifford on $n+m$ qubits for any positive integer $m$. 
\end{lemma}

\begin{proof}
We show the contrapositive. 
Suppose $U'=U\otimes I_{2^m}$ is semi-Clifford. 
Consider the subgroup $G$ of $\cP_{n+m}$ consisting of all $P$ such that $U'P(U')^{-1} \in \cP_{n+m}$. 
We know by~\Cref{lem:sc-with-subgps} that $G$ contains a maximal abelian subgroup $A$ of $\cP_{n+m}$. 
Let $B=\la Z_{n+1}, \ldots, Z_{n+m} \ra \subseteq G$. 
Using~\Cref{new-max-ab-subgp}, we get a maximal abelian subgroup $A'$ of $\cP_{n+m}$ such that $B\subseteq A' \subseteq \la A, B \ra \subseteq G$. 
Thus, there exists a subgroup $A_1$ of $\cP_n$ with $A'= \la A_1, B \ra$. 
We can see that $A_1$ must have at least $2^n$ elements up to phase, so it must be a maximal abelian subgroup of $\cP_n$. 
So $UA_1U^{-1} \subseteq \cP_n$, which means that $U$ is semi-Clifford.
\end{proof}

It follows from~\Cref{prop:u3} and~\Cref{lem:gottesman-mochon} that $\cC_3$ contains a non-semi-Clifford permutation gate supported on $n$ qubits for any $n \geq 7$.
As mentioned in~\Cref{sec:prelim-CH}, it has been proved that all gates in $\cC_3$ supported on at most $n = 4$ qubits are semi-Clifford, while the cases for $n = 5$ and $6$ remain open.
We now show that all gates in $\cC_3^{\sym}$ on at most $6$ qubits is semi-Clifford.

\begin{proposition}
\label{prop:staircase-form-semi-clifford}
Consider a staircase form permutation $\pi \in \cC_3$ and the descending multiplication corresponding to $\pi$. The following are equivalent:

\begin{enumerate}
\item $\pi$ is semi-Clifford.

\item $\pi^{-1} \in \cC_3$.

\item All polynomials in the polynomial form of $\pi$ have degree at most $2$.

\item For all indices $i,j,k$, we have $e_ie_je_k=0$.
\end{enumerate}
\end{proposition}

\begin{proof}

We have (1) implies (2) by \Cref{lem:respect-sc}, and (2) implies (3) by \Cref{ck-perm-poly} applied to $\pi^{-1}$. Also, (3) implies there are no degree $3$ terms in any polynomial in the polynomial form of $\pi$, so $e_ie_je_k=0$ for all distinct $i,j,k$ by \Cref{prod-ei-and-ai}. This implies that $e_ie_je_k=0$ for all $i,j,k$ since $e_l^2=0$ for all $l$, and multiplication is commutative. Therefore (3) implies (4).

It remains to prove (4) implies (1). Suppose (4) holds, which implies that any product of at least three elements is zero. Define $v_{ij}$ to be such that $\pi(e_i+e_j+v_{ij})=e_i+e_j$.
For any $i<j$ and any $w$, we have by \Cref{eq:pi(vij)} that $\pi(v_{ij}) = e_ie_j$, so by \Cref{pi-v-w} we have that 
\begin{align}\label{eq:apdx-vijw}
    \pi(v_{ij}+w) = \pi(v_{ij})+\pi(w)+\pi(v_{ij})\pi(w) = e_ie_j + \pi(w) + e_ie_j\pi(w) = e_ie_j + \pi(w),
\end{align}
as the product of any three elements is zero.
Let $P=\Span\{v_{ij}: i < j\}$ and let $Q=\Span\{e_ie_j: i < j\}$.
Observe that $vw \in Q$ for all $v$ and $w$. 
Observe that repeated use of \Cref{eq:apdx-vijw}, together with the fact that $\pi(0)=0$, implies that $\pi(p) \in Q$ for all $p \in P$. 
Since $\dim(P) \leq n$, let $k=n-\dim(P)$. 
Take a basis of $P$, and label its elements as $p_{k+1}, \dots, p_n$. 
Extend this basis to a basis $p_1, \dots, p_n$ of $\mathbb{F}_2^n$. Let $q_i = \pi(p_i)$ for all $i$, so $q_i \in Q$ for all $i>k$.
Again by repeated use of \Cref{eq:apdx-vijw}, we have for all $w$
\begin{align*}
    \pi(p_i + w) = q_i + \pi(w).
\end{align*}

The intuition for what follows is that $q_i$ will be a basis for $\mathbb{F}_2^n$, and if we change the bases of the inputs and outputs of $\pi$ from $e_i$ to, respectively, $p_i$ and $q_i$, then $\pi$ becomes a mismatch-free product of Toffolis, where $p_1, \dots, p_k$ are the controls and $p_{k+1}, \dots, p_n$ are the targets. While we will not prove this statement explicitly, our proof is guided by this intuition.

\begin{claim}
    $q_{k+1}, \dots, q_n$ is a basis of $Q$, and $q_1, \dots, q_n$ is a basis of $\mathbb{F}_2^n$.
\end{claim}
\begin{proofclaim}
Note that for any $S \subseteq \{k+1, \dots, n\}$, $\pi(\sum_{i \in S} p_i) = \sum_{i \in S} q_i$. If $\sum_{i \in S} q_i = 0$, then $\pi(\sum_{i \in S} p_i) = 0$, which means $\sum_{i \in S} p_i = 0$. This is a contradiction with the fact that $p_i$ is a basis for $P$. Therefore $q_{k+1}, \dots, q_n$ are linearly independent.
To see that they span $Q$, note that for any $i<j$, we can take $S \subseteq \{k+1, \dots, n\}$ such that $\sum_{m \in S} p_m = v_{ij}$. 
Then $\sum_{m \in S} q_m = e_ie_j$, which means $e_ie_j \in \Span(q_{k+1}, \dots, q_n)$. 
Thus $Q \subseteq \Span(q_{k+1}, \dots, q_n)$ and $q_{k+1}, \dots, q_n$ is a basis of $Q$.

To show that $q_1, \dots, q_n$ is a basis of $\mathbb{F}_2^n$, suppose $S \subseteq \{1, \dots, n\}$ is such that $\sum_{i \in S} q_i = 0$, we will show $S = \varnothing$. 
Let $A = S \cap \{1, \dots, k\}$ and $B = S \cap \{k+1, \dots, n\}$, so that $A \sqcup B = S$. 
Then 
\begin{align*}
    \sum_{i \in A} q_i = \sum_{i \in B} q_i =  \pi \left( \sum_{i \in B} p_i \right) \in Q,
\end{align*}
because $\sum_{i \in B} p_i \in P$.
Let $a = \sum_{i \in A} p_i$. 
By repeatedly applying \Cref{pi-v-w}, we have 
\begin{align*}
    \pi(a) = \sum_{T \subseteq A, T \neq \varnothing} \prod_{i \in T} q_i = \sum_{i \in A} q_i + \sum_{T \subseteq A, |T| \geq 2} \prod_{i \in T} q_i.
\end{align*}
We know $\sum_{i \in A} q_i \in Q$, and $\prod_{i \in T} q_i \in Q$ for any $T \subseteq A$. 
Therefore $\pi(a)$ is in $Q$, and we can take $C \subseteq \{k+1, \dots, n\}$ so that $\pi(a) = \sum_{i \in C} q_i$. 
But we also have $\pi(\sum_{i \in C} p_i) = \sum_{i \in C} q_i$, therefore 
\begin{align*}
    \sum_{i \in A} p_i = a = \sum_{i \in C} p_i.
\end{align*}
However, $A \cap C = \varnothing$ and the $p_i$ form a basis, therefore we must have that $A=C=\varnothing$. This implies
\begin{align*}
    0 = a = \sum_{i \in B} q_i = \pi(\sum_{i \in B} p_i),
\end{align*}
which means $\sum_{i \in B} p_i = 0$, so $B = \varnothing$. 
We conclude that $q_1, \dots, q_n$ are linearly independent, and thus are a basis of $\mathbb{F}_2^n$.
\end{proofclaim}

Now we can take invertible linear maps $\xi_1$ and $\xi_2$ so that $\xi_1(q_i)=e_i$ and $\xi_2(e_i)=p_i$ for all $i$. Note that $\xi_1$ and $\xi_2$ are Clifford permutations by \Cref{clifford-perm}. Let $\mu = \xi_1\pi\xi_2$.  
All polynomials in the polynomial forms of $\xi_1$ and $\xi_2$ have degree at most $1$, and all polynomials in the polynomial form of $\pi$ have degree at most $2$ by \Cref{prod-ei-and-ai} (since any product of at least three $e_i$ is zero), so all polynomials in the polynomial form of $\mu$ have degree at most $2$. Also, by construction, $\mu(e_i)=e_i$ for all $i$, and $\mu(0)=0$. Then the polynomial form of $\mu$ can be written as $(a_1, \dots, a_n) \mapsto (a_1+r_1, \dots, a_n+r_n)$ for some quadratic polynomials $r_i$.

We know $\xi_1(v+w) = \xi_1(v)+\xi_1(w)$ and $\xi_2(v+w) = \xi_2(v)+\xi_2(w)$ for all $v, w$ because $\xi_1, \xi_2$ are linear maps. 
For all $i<j$, by \Cref{pi-v-w} we can write 
\begin{align*}
    \mu(e_i+e_j) = \xi_1\pi(p_i+p_j) = \xi_1(q_i+q_j+q_iq_j) = e_i+e_j+\xi_1(q_iq_j).
\end{align*}
Note that $q_iq_j \in Q$, as the product of any two elements is in $Q$.
Therefore we can take $S_{ij} \subseteq \{ k+1, \dots, n\}$ with $\sum_{m \in S_{ij}} q_m = q_iq_j$, which means 
\begin{align*}
    \mu(e_i+e_j) = e_i+e_j+\xi_1(q_iq_j) = e_i+e_j+\sum_{m \in S_{ij}} e_m.
\end{align*}
This implies that, for all $m=1, \dots, n$, $r_m$ contains an $a_ia_j$ term if and only if $m \in S_{ij}$.
In particular, for all $m \leq k$, $r_m$ does not contain an $a_ia_j$ term for all $i<j$.
Therefore $r_m=0$ for all $m \leq k$, which means $\mu$ commutes with $Z_m$ for all $m \leq k$.
Complimentarily, for all $i>k$ and all $w$, we have 
\begin{align*}
    \mu(e_i+w) = \xi_1 \pi \xi_2 (e_i+w) = \xi_1 \pi (p_i+\xi_2(w)) = \xi_1(q_i + \pi \xi_2(w)) = \xi_1(q_i) + \xi_1\pi\xi_2(w) = e_i + \mu(w).
\end{align*}
This yields that $\mu$ commutes with $X_i$ for all $i \geq k$.

Since $\mu$ commutes with all of $Z_1, \dots, Z_k, X_{k+1}, \dots, X_n$, it also commutes with all elements of $\langle Z_1, \dots, Z_k, X_{k+1}, \dots, X_n \rangle$, which is a maximal abelian subgroup of $\mathcal{P}_n$. 
Thus $\mu$ is semi-Clifford by \Cref{lem:sc-with-subgps}, which means $\pi=\xi_1^{-1} \mu \xi_2^{-1}$ is semi-Clifford since $\xi_1^{-1}$ and $\xi_2^{-1}$ are Clifford. This completes our proof.
\end{proof}

\begin{theorem} \label{seven-is-best}
The smallest number of qubits for which there exists a non-semi-Clifford permutation in $\cC_3$ is $7$. 
\end{theorem}

\begin{proof}
We already know from \Cref{subsec:family_example_k3} that there exists a non-semi-Clifford permutation in $\cC_3$ on $7$ qubits. Let us show this is minimal. Suppose $\pi$ is a non-semi-Clifford permutation in $\cC_3$ on $n$ qubits. Write $\pi = \phi_1\mu\phi_2$ as in \Cref{thm:c3-staircase}. If $\mu$ were semi-Clifford, then $\pi$ would be semi-Clifford; thus $\mu$ is non-semi-Clifford. Then \Cref{prop:staircase-form-semi-clifford} implies that, in the descending multiplication corresponding to $\mu$, there exist $i,j,k$ with $e_ie_je_k \neq 0$. These must be pairwise distinct (as $e_m^2=0$ for all $m$); then, applying \Cref{prod-ei-nonzero} with $S=\{i,j,k\}$ implies that $n \geq 2^3-1 = 7$, as desired. The proof is complete.
\end{proof}

\begin{remark}
\Cref{seven-is-best} can alternatively be shown by computer search, as shown as Theorem 5.4 in the previous version of this paper~\cite{he2024permutation}. \Cref{thm:c3-staircase} is key to the computer search: it suffices to check only staircase form permutations in $\cC_3$, and there are a total of $2^{\binom{6}{3}} = 1,048,576$ staircase form permutations on six qubits (whether in $\cC_3$ or not), which is a reasonable number for a computer search. The relevant code can be found at \url{https://github.com/Likable-outlieR/clifford-hierarchy}.
\end{remark}

\section{Semi-Clifford permutations} \label{sec-sc-perm}

In this appendix, we show that semi-Clifford permutation gates are in correspondence with mismatch-free circuits of multi-controlled NOT gates.
Recall that we use $C^kX$ to denote a NOT gate with $k$ control qubits.
Let $C^*X = \{C^kX, k\geq 0\}$ denote the collection of all multi-controlled NOT gates.

\begin{lemma} \label{mismatch-free-is-sc}
Any mismatch-free product $\mu$ of $C^*X$ gates is semi-Clifford.
\end{lemma}

\begin{proof}
Consider an $X$ gate on every target qubit and a $Z$ gate on every non-target qubit. These gates generate a maximal abelian subgroup of $\cP_n$ up to phase, and $\mu$ will commute with every element of this subgroup. The claim follows from~\Cref{lem:sc-with-subgps}.
\end{proof}

\begin{theorem} \label{thm:semi-clifford-mismatch-free}
For any permutation gate $\pi$ that is semi-Clifford, there exist Clifford permutations $\phi_1, \phi_2$ and a mismatch-free product $\mu$ of $C^*X$ gates such that $\pi = \phi_1 \mu \phi_2$.
\end{theorem}

Before proving this theorem, we prove a few useful lemmas. 
Given a vector $u\in \mathbb{F}_2^n$, we use the notation $X^u$ to denote the operator $X^{u[1]}\otimes X^{u[2]}\otimes \dots \otimes X^{u[n]}$, where $u[i]$ denote the $i$-th index of $u$, $X^1 = X$ and $X^0 = I$. 
Define $Z^u$ similarly. 
Any Pauli operator $P\in \cP_n$ has a decomposition as $P = cX^uZ^v$ for some phase $c$ and $u, v\in \mathbb{F}_2^n$. 

\begin{lemma} \label{pauli-is-pd}
Every Pauli gate can be uniquely written as the product of a permutation gate and a diagonal gate; furthermore, the permutation gate and diagonal gate are each individually Pauli.
\end{lemma}

\begin{proof}
Any Pauli operator $P$ can be written as $P = cX^uZ^v$, where $p = X^u$ is a permutation gate and $d = cZ^v$ is a diagonal gate. 
It remains to show uniqueness of the representation. To see this, suppose $P=p'd'$ for permutation $p'$ and diagonal $d'$. 
We have $(p')^{-1}p = (p')^{-1}Pd^{-1} = d'd^{-1}$. 
Since $(p')^{-1}p$ is a permutation matrix, $d'd^{-1}$ is a diagonal matrix, and the only diagonal permutation matrix is the identity, we must have $(p')^{-1}p = d'd^{-1} = I$. Therefore $p'=p$ and $d'=d$, as desired.
\end{proof}

The following lemma is a special case of~\Cref{thm:semi-clifford-mismatch-free}.

\begin{lemma}\label{phi-mu}
Suppose $\pi$ is a permutation gate, and $m \leq n$ is a nonnegative integer such that \\ $\pi X_1 \pi^{-1}, \ldots, \pi X_m \pi^{-1}, \pi Z_{m+1} \pi^{-1}, \ldots, \pi Z_n \pi^{-1} \in \cP_n$. Then there exist Clifford permutations $\phi_1, \phi_2$ and a mismatch-free product $\mu$ of $C^*X$ gates such that $\pi = \phi_1 \mu \phi_2$. 
\end{lemma}
\begin{proof}
Let $X'_i = \pi X_i \pi^{-1}$. By~\Cref{conj-cliff-perm}
we can take a Clifford permutation $\nu$ such that $X'_i = \nu X_i \nu^{-1}$. 
Replacing $\pi$ with $\nu^{-1} \pi$, which preserves the property that $\pi Z_j \pi^{-1} \in \cP_n$ for $m+1\le j\le n$, we can assume without loss of generality that $\pi$ commutes with $X_1, \ldots, X_m$. 

For $m+1\le j\le n$, $\pi Z_j \pi^{-1}$ is a diagonal gate in the Pauli group, i.e.\ $\epsilon_j Z^{w_j}$ for some vector $w_j$ and $\epsilon_j=\pm 1$. 
Since $\pi Z_j \pi^{-1}$ commutes with $\pi X_i \pi^{-1} = X_i$ for $i\in [m]$, we must have $w_j$ is zero on the first $m$ components for all $m+1\le j\le n$. 
Let $\chi$ be the product of $X_j$ over all $j$ with $\epsilon_j = -1$. By replacing $\pi$ with $\pi\chi$, we can assume without loss of generality 
that $\epsilon_j = 1$ for all $j$, while preserving the property that $\pi$ commutes with $X_1, \ldots, X_m$, and without changing $w_{m+1}, \ldots, w_n$.

Since $Z_j$ are independent, the vectors $w_j$ are also linearly independent. 
Since the first $m$ components of each $w_j$ are zeros, $e_1, \ldots, e_m, w_{m+1}, \ldots, w_n$ forms a linear basis. 
Hence, there exists an invertible matrix $M$ with $Me_i=e_i$ for $i\in [m]$ and $Me_j = w_j$ for $m+1\le j\le n$. 
Consider the map $\varpi$ defined as $|v \ra \mapsto |M^\top v \ra$ which is a Clifford permutation by~\Cref{clifford-perm}. 
Then, $\varpi (\pi Z_j \pi^{-1}) \varpi^{-1} = \varpi Z^{Me_j} \varpi^{-1}$ sends 
\[ 
|v \ra \mapsto |(M^\top )^{-1}v \ra 
\mapsto (-1)^{v^\top M^{-1}Me_j} |(M^\top )^{-1} v \ra 
\mapsto (-1)^{v^\top e_j} |v \ra, 
\] 
so $\varpi \pi Z_j \pi^{-1} \varpi^{-1} = Z_j$. Also, since the first $m$ components of $w_j$ are zeros for $m+1\le j\le n$, we have $M^\top e_i = e_i$ for $i = 1, \ldots, m$.
Therefore, $\varpi \pi X_i \pi^{-1} \varpi^{-1} = \varpi X_i \varpi^{-1}$ sends
\[ 
|v \ra \mapsto |(M^\top )^{-1} v \ra \mapsto | (M^\top )^{-1}v + e_i \ra \mapsto |M^\top ((M^\top )^{-1}v+e_i) \ra 
= |v+e_i \ra, 
\]
so $\varpi \pi X_i \pi^{-1} \varpi^{-1} = X_i$.
Since $\varpi$ is a Clifford permutation and $\varpi \pi$ commutes with $X_1, \ldots, X_m$ and $Z_{m+1}, \ldots, Z_n$, by replacing $\pi$ with $\varpi \pi$, we can assume without loss of generality that $\pi$ commutes with $X_1, \ldots, X_m, Z_{m+1}, \ldots, Z_n$. 

Now consider the polynomial representation $(\pi_1, \ldots, \pi_n)$ of $\pi$. 
For $m+1\le j\le n$, $\pi Z_j = Z_j \pi$ implies that $\pi_j(v) = v_j$ for $v \in \bF_2^n$. 
For $1\leq j\leq m$, $\pi X_j = X_j \pi$ implies that
$\pi(v+e_j) = \pi(v) + e_j$, i.e.\ $\pi_i(v+e_j) = \pi_i(v)$ for $i\neq j$, and $\pi_j(v+e_j) = \pi_j(v) + 1$. 
So $\pi_j$ is $v_j$ plus a polynomial $p_j$ in terms of only $v_{m+1}, \ldots, v_n$. 

Note that every monomial in $p_j$ corresponds to a $C^*X$ gate with qubit $j$ as target and a subset of qubits in $\{m+1, \ldots, n\}$ as controls. 
Now $\pi$ is the product of all these $C^*X$ gates and is mismatch-free, since qubits $1, \ldots, m$ are used only as targets and qubits $m+1, \ldots, n$ are never used as targets.
\end{proof}

We now prove~\Cref{thm:semi-clifford-mismatch-free} by reducing to the case of~\Cref{phi-mu}.
\begin{proof}[Proof of~\Cref{thm:semi-clifford-mismatch-free}]
Let $G$ be the subgroup of $\cP_n$ of all elements $P$ with $\pi P \pi^{-1} \in \cP_n$, and let $M$ be the set consisting of all permutations in $G$; 
now $M = G \cap \cX$, so $M$ is an abelian subgroup of $G$. 
By~\Cref{conj-cliff-perm} we can find a Clifford permutation $\nu$ such that $M = \nu \la X_1, \ldots, X_m \ra \nu^{-1}$ for some $m$.
If we replace $\pi$ by $\pi\nu$, we would replace $G$ with $\nu^{-1}G \nu$ and replace $M$ with 
\[
\nu^{-1}G \nu \cap \cX = \nu^{-1}G \nu \cap \nu^{-1} \cX \nu = \nu^{-1} M \nu = \la X_1, \ldots, X_m \ra.
\]
Therefore, let us assume without loss of generality that $M = \la X_1, \ldots, X_m \ra$ for some $m$.

We know $G$ contains a maximal abelian subgroup $A$ of $\cP_n$, since $\pi$ is semi-Clifford. 
Applying~\Cref{new-max-ab-subgp} on $A$ and $M$, we get a maximal abelian subgroup $A'$ of $\cP_n$ with $M \subseteq A' \subseteq \la A, M \ra \subseteq G$. 
We claim that $A' = \la X_1, \ldots, X_m, Z_{m+1}, \ldots, Z_n\ra$ up to phase.
To see this, take a basis $X_1, \ldots, X_m, W_{m+1}, \ldots, W_n$ for $A'$.
Decompose $W_i = c_iX^{u_i}Z^{v_i}$.
We can assume the first $m$ indices of $u_i$ are zeros. 
Since $W_i$ commutes with $X_1, \ldots, X_m$, the first $m$ indices of $v_i$ must be zeros. 
It now suffices to show that $u_i = 0$ for all $m < i \le n$.

Let $p = X^{u_i}$ and $d = c_iZ^{v_i}$.
Since $W_i\in A' \subseteq G$, we have $W_i' = \pi W_i\pi^{-1}\in \cP_n$.
Note that 
\[
W_i' = \pi pd \pi^{-1} = (\pi p \pi^{-1})(\pi d \pi^{-1}),
\]
where $\pi p \pi^{-1}$ is a permutation and $\pi d \pi^{-1}$ is diagonal. 
It follows from~\Cref{pauli-is-pd} that this decomposition is unique and $\pi p \pi^{-1}, \pi d \pi^{-1} \in \cP_n$. 
So $p, d \in G$.  
Since $p \in \cX$, we have $p \in G \cap \cX = M = \la X_1, \ldots, X_m \ra$. 
In other words, $u_i[j] = 0$ for all $j > m$. 
Therefore, we have $u_i=0$ and $A' = \la X_1, \ldots, X_m, Z_{m+1}, \ldots, Z_n\ra$ up to phase.
The theorem then follows from~\Cref{phi-mu}. 
\end{proof}

We conclude with the following characterization of semi-Clifford permutation gates. 

\begin{theorem} \label{sc-level}
For any positive integer $k$, a permutation gate $\pi$ is a semi-Clifford gate in $\cC_{k+1}$ if and only if there exist Clifford permutations $\phi_1, \phi_2$ and a mismatch-free product $\mu$ of $C^*X$ gates such that $\pi = \phi_1 \mu \phi_2$ and, in $\mu$, each gate has at most $k$ controls.
\end{theorem}

\begin{proof} 
For the ``if'' direction, $\pi$ being semi-Clifford follows from~\Cref{mismatch-free-is-sc}, and $\pi\in \cC_{k+1}$ follows directly from~\Cref{thm:anderson-mismatch-free} (along with part 2 of~\Cref{prop:cliff_hier}).

For the ``only if'' direction, we apply~\Cref{thm:semi-clifford-mismatch-free} to get a representation $\phi_1 \mu \phi_2$ where $\phi_1$ and $\phi_2$ are Clifford permutations and $\mu$ is a mismatch-free product of $C^*X$ gates. As any two gates in such a product commutes, and every such gate is its own inverse, we can assume without loss of generality that no gate is repeated. 
By part 2 of~\Cref{prop:cliff_hier}, $\mu \in \cC_{k+1}$. 
By~\Cref{ck-perm-poly}, in the polynomial representation of $\mu^{-1}$, every coordinate has degree at most $k$. 
Note that $\mu^{-1} = \mu$. 
If there is a gate in $\mu$ with $m>k$ controls, this would yield a monomial of degree $m$ in $\mu$ which would not be canceled out.
Therefore every gate in $\mu$ has at most $k$ controls, as desired.
\end{proof}

Our result has two immediate corollaries. 

\begin{corollary} \label{sc-perms-in-c3}
A permutation gate $\pi$ is a semi-Clifford gate in $\cC_3$ if and only if there exist Clifford permutations $\phi_1, \phi_2$ and a mismatch-free product $\mu$ of Toffoli gates such that $\pi = \phi_1 \mu \phi_2$.
\end{corollary}

\begin{corollary}
Every semi-Clifford permutation gate is in $\cC_n$. 
\end{corollary}
\begin{proof}
The claim follows from~\Cref{thm:semi-clifford-mismatch-free,sc-level} and the fact that a $C^*X$ gate on $n$ qubits has at most $n-1$ controls. 
\end{proof}

\end{document}